\documentclass[a4paper,12pt]{article}

\usepackage[utf8]{inputenc}
\usepackage[T2A]{fontenc}
\usepackage{amssymb}
\usepackage{amsmath}
\usepackage{mathrsfs}
\usepackage{euscript}
\usepackage[russian,english]{babel}
\usepackage{array}
\usepackage{theorem}
\usepackage[ruled]{algorithm}
\usepackage[noend]{algorithmic}
\usepackage{graphicx}
\usepackage{color}
\usepackage{ifthen}
\usepackage{url}
\usepackage{makeidx}
\usepackage{pb-diagram}
\usepackage{listings}
\usepackage{fancyhdr}
\usepackage{misccorr}
\usepackage{indentfirst}
\usepackage{tocloft}
\usepackage{soul}
\usepackage{nomencl}
\usepackage{calc}
\usepackage[all]{xy}
\usepackage{cite}

\usepackage[left=1cm, right=1cm, top=1cm, bottom=1.5cm]{geometry}

\renewcommand{\geq}{\geqslant}
\renewcommand{\leq}{\leqslant}
\renewcommand{\ge}{\geqslant}

\renewcommand{\emptyset}{\varnothing}

\newcommand\myop[1]{\mathop{\operator@font #1}\nolimits}
\newcommand\mylim[1]{\mathop{\operator@font #1}\limits}

\newenvironment{enumerate*}
{%
    \begingroup
    \renewcommand{\@listi}{%
        \topsep=\smallskipamount
        \parsep=0pt
        \parskip=0pt
        \itemsep=0pt
        \itemindent=0ex
        \labelsep=1.5ex
        \leftmargin=7ex
        \rightmargin=0pt} 
    \begin{enumerate}%
}{%
    \end{enumerate}%
    \endgroup
}

\theoremstyle{plain}
\gdef\th@plain{\normalfont
    \def\@begintheorem##1##2{%
        \item[\hskip\labelsep\theorem@headerfont \hspace*{5mm} ##1\ ##2. ]}%
    \def\@opargbegintheorem##1##2##3{%
        \item[\hskip\labelsep\theorem@headerfont \hspace*{5mm} ##1\ ##2 {\normalfont ##3.} ]}%
}

\theorempreskipamount=\smallskipamount
\theorempostskipamount=\smallskipamount
\theorembodyfont{\rmfamily}

\newtheorem{State-rm}{Предложение}

\def\qed{\rule{6pt}{6pt}}

\newtheorem{lemma}{Lemma}

\newtheorem{remark}{Remark}

\newenvironment{proof}[0]{%
    {\itshape\bfseries  Proof.}
}
{%
    \qed
    \vskip 7pt
}

    {\begin{algorithm}[#1]\begin{algorithmic}[1]%
        \renewcommand{\ALC@it}{%
            \refstepcounter{ALC@line}%
            \addtocounter{ALC@rem}{1}%
            \ifthenelse{\equal{\arabic{ALC@rem}}{1}}{\setcounter{ALC@rem}{0}}{}%
            \item}}%
    {\end{algorithmic}\end{algorithm}}

\renewcommand\emph[1]{\textit{#1}}

\usepackage{hyperref}
\usepackage{enumitem}
\usepackage{algorithmic}

\usepackage{tikz}
\usetikzlibrary{shapes.geometric, arrows}

\tikzstyle{box} = [rectangle, rounded corners, minimum width=1cm, minimum height=1cm,text centered, draw=black, fill=yellow!30]
\tikzset{oplus/.style={path picture={%
    \draw[black]
    (path picture bounding box.south) -- (path picture bounding box.north)
    (path picture bounding box.west) -- (path picture bounding box.east);
}}}
\tikzstyle{xor} = [oplus, draw = black, circle, minimum size = 0.5cm]

\newcounter{figureCOUNTER}

\newcommand{\fd}[1]{\mathbb{F}_{#1}}

\newcommand{\fdo}[1]{\widehat{\mathbb{F}}_{#1}}

\newcommand{\bobsym}[2]{B_{#2}(#1)}
\newcommand{\fdn}[1]{B(\mathbb{F}_{#1})}
\newcommand{\bobset}[1]{B^0_{#1}}
\newcommand{\bobseta}[1]{B^1_{#1}}

\newcommand{\sdsbox}[1]{S(#1)}
\newcommand{\sdfunc}[1]{H(#1)}

\usetikzlibrary{quantikz}

\newsavebox{\boxX}
\newsavebox{\boxZ}
\newsavebox{\boxH}
\newsavebox{\boxCnot}
\newsavebox{\boxSwap}
\newsavebox{\boxToffoli}

\usepackage{array}

\allowdisplaybreaks

\providecommand{\keywords}[1]{
	\medskip
    \noindent
	\small\textbf{Keywords:} #1
	}

\providecommand{\affili}[1]{
	\vspace{-1cm}
	\begin{center} 
	{\it \small #1}
	\end{center}
	\vspace{.1cm}}

\providecommand{\email}[1]{
	\vspace{-0.5cm}
	\begin{center} 
	{\tt \small #1}
	\end{center}
	\vspace{.1cm}}

\begin{document}
	
\title{Mathematical problems and solutions of the Ninth International Olympiad in Cryptography NSUCRYPTO}
	
\author{V.\,A.~Idrisova$^{1}$, N.\,N.~Tokareva$^{1}$, A.\,A.~Gorodilova$^{1}$, I.\,I.~Beterov$^{2}$, T.\,A.~Bonich$^{1}$,\\ E.\,A.~Ishchukova$^{3}$, N.\,A.~Kolomeec$^{1}$, A.\,V.~Kutsenko$^{1}$, E.\,S.~Malygina$^{4}$,\\ I.\,A.~Pankratova$^{5}$, M.\,A.~Pudovkina$^{6}$, A.\,N.~Udovenko$^{7}$}

\date{}

\maketitle

\affili{$^{1}$Novosibirsk State University, Novosibirsk, Russia\\
              $^{2}$Rzhanov Institute of Semiconductor Physics, Novosibirsk, Russia\\
              $^{3}$Southern Federal University, Rostov-on-Don, Russia\\
              $^{4}$Immanuel Kant Baltic Federal University, Kaliningrad, Russia\\
              $^{5}$Tomsk State University, Tomsk, Russia\\
              $^{6}$National Research Nuclear University MEPhI, Moscow, Russia\\
              $^{7}$CryptoExperts, Paris, France
              } 

\email{vvitkup@yandex.ru, crypto1127@mail.ru, gorodilova@math.nsc.ru, beterov@isp.nsc.ru,\\ tatianabonich@yandex.ru, uaishukova@sfedu.ru, kolomeec@math.nsc.ru, alexandrkutsenko@bk.ru, EMalygina@kantiana.ru, pank@mail.tsu.ru, maricap@rambler.ru, aleksei.udovenko1@gmail.com}

\abstract{Every year the International Olympiad in Cryptography Non-Stop University CRYPTO (NSUCRYPTO) offers mathematical problems for university and school students and, moreover, for professionals in the area of cryptography and computer science. The mail goal of NSUCRYPTO is to draw attention of students and young researchers to modern cryptography and raise awareness about open problems in the field. We present problems of NSUCRYPTO'22 and their solutions. There are 16 problems on the following topics: ciphers, cryptosystems, protocols, e-money and cryptocurrencies, hash functions, matrices, quantum computing, S-boxes, etc. They vary from easy mathematical tasks that could be solved by school students to open problems that deserve separate discussion and study. So, in this paper, we consider several open problems on three-pass protocols, public and private keys pairs, modifications of discrete logarithm problem, cryptographic permutations and quantum circuits.

}

\section{Introduction}

{\bf Non-Stop University CRYPTO} ({\bf NSUCRYPTO}) is the unique interna\-tional competition for professionals, school and university students, providing various problems on theoretical and practical aspects of modern cryptography, see \cite{nsucrypto}. The main goal of the olympiad is to draw attention of young researchers not only to competetive fascinating tasks, but also to sophisticated and tough scientific problems at the intersection of mathematics and cryptography. That is why each year there are several open problems in the list of tasks that require rigorous studying and deserve a separate publication in case of being solved. Since NSUCRYPTO holds via the Internet, everybody can easily take part in it.  Rules of the Olympiad, the archive of problems, solutions and many more can be found at the official website \cite{nsucrypto-rules}. 

The first Olympiad was held in 2014, since then more than 3000 students and specialists from almost 70 countries took part in it. The Program committee now is including 22 members from cryptographic groups all over the world. Main organizers and partners are Cryptographic Center (Novosibirsk), Mathematical Center in Akademgorodok, Novosibirsk State University, KU Leuven, Tomsk State University, Belarusian State University, Kovalevskaya North-West Center of Mathematical~Research and Kryptonite.

This year 37 participants in the first round and 27 teams in the second round from 14 countries became the winners (see the list \cite{nsucrypto-winners}). This year we proposed 16 problems to participants and 5 of them were entirely open or included some open questions. Totally, there were 623 particpants from 36 countries. 

Following the results of each Olympiad we also publish scientific articles with detailed solutions and some analysis of the solutions proposed by the participants, including advances on unsolved ones, see \cite{nsucrypto-2014, nsucrypto-2015,  nsucrypto-2016, nsucrypto-2017, nsucrypto-2018, nsucrypto-2019, nsucrypto-2020, nsucrypto-2021}.

\section{An overview of open problems}

One of the main characteristic of the Olympiad is that unsolved scientific problems are proposed to the participants in addition to problems with known solutions. All 31 open problems that were offered since the first NSUCRYPTO can be found here \cite{nsucrypto-unsolved}. Some of these problems are of great interest to cryptographers and mathematicians for many years. These are such problems as ``APN permutation'' (2014), ``Big Fermat numbers'' (2016), ``Boolean hidden shift and quantum computings'' (2017), ``Disjunct Matrices'' (2018), and others. 

Despite that it is marked that the problem is open and therefore it requires a lot of hard work to advance, some of the problems we suggested are solved or partially solved by our participants during the Olympiad. For example, problems ``Algebraic immunity'' (2015), ``Sylvester matrices'' (2018), ``Miller~---~Rabin revisited'' (2020) were solved completely. 
Also, partial solutions were suggested for problems ``Curl27'' (2019), ``Bases'' (2020),  ``Quantum error correction'' (2021) and ``s-Boolean sharing'' (2021). 

Moreover, some researchers continue to work on solutions even after the Olym\-piad was over. For example, authors of~\cite{20-Kiss-Nagy} proposed a complete solution for problem ``Orthogonal arrays'' (2018). Partial solutions for another open problem, ``A secret sharing'', (2014) were presented in~\cite{17-Geut},~\cite{19-Geut}, and a recursive algorithm for finding the solution was proposed in~\cite{19-Ayat}.

This year, two open problems ware solved during the Olympiad. These are problems ``Public keys for e-coins'' (see Problem \hyperlink{pr-publickeys}{4.10}) and ``Quantum entanglement'' (see Problem \hyperlink{pr-entanglement}{4.16}).

\section{Problem structure of the Olympiad}
\label{problem-structure}

There were 16 problems stated during the Olympiad, some of them were included in both rounds (Tables\;1,\,2).
Section A of the first round consisted of six problems, while Section B of the first round consisted of eight problems. The second round was composed of eleven problems; five of them included unsolved questions (awarded special prizes). 

\begin{table}[H]
\centering\footnotesize
\begin{tabular}{cc}
\begin{tabular}{|c|l|c|}
  \hline
  N & Problem title & Max score \\
  \hline
  \hline
  1 & \hyperlink{pr-numberspoints}{Numbers and points} & 4 \\
    \hline
  2 & \hyperlink{pr-wallet}{Wallets}  & 4 \\
    \hline
  3 &  \hyperlink{pr-long}{A long-awaited event} & 4\\
    \hline
  4 &  \hyperlink{pr-hiddenprimes}{Hidden primes} & 4 \\
  \hline
    5 & \hyperlink{pr-facetoface}{Face-to-face} & 4 \\
 \hline
  6 &  \hyperlink{pr-cryptolocks}{Crypto locks} & 4 + open problem\\
    \hline

\end{tabular}

&

\begin{tabular}{|c|l|c|}
  \hline
  N & Problem title & Max score \\
  \hline
    \hline
  1 & \hyperlink{pr-numberspoints}{Numbers and points} & 4 \\
  \hline
 2 &  \hyperlink{pr-hiddenprimes}{Hidden primes} & 4\\
   \hline
   3 & \hyperlink{pr-facetoface}{Face-to-face} & 4 \\
 \hline
 4    & \hyperlink{pr-matrixreduction}{Matrix and reduction} & 4 \\
  \hline 
  5 & \hyperlink{pr-reversingagate}{Reversing a gate} & 6 \\
    \hline
  6 & \hyperlink{pr-bob}{Bob's symbol} & 8 \\
   \hline
  7 & \hyperlink{pr-cryptolocks}{Crypto locks} & 4 + open problem \\
    \hline
  8 & \hyperlink{pr-publickeys}{Public keys for e-coins} & open problem \\
    \hline
\end{tabular}
\\
{\bf Section A}   &  {\bf Section B}\\
\end{tabular}

\medskip


\caption{{Problems of the first round} }

\end{table}

\vspace{-0.6cm}

\begin{table}[H]
\centering\footnotesize
\begin{tabular}{|c|l|c|}
  \hline
  N & Problem title & Max score \\
  \hline
    \hline
  1 &  \hyperlink{pr-cpproblem}{CP problem} & open problem\\
    \hline
  2 &  \hyperlink{pr-interpolation}{Interpolation with errors} & 8\\
    \hline
  3 &  \hyperlink{pr-HAS01}{HAS01} & 8 \\
    \hline 
  4 &  \hyperlink{pr-PHIGFS}{Weaknesses of the PHIGFS} & 8 \\
    \hline
  5 &  \hyperlink{pr-SBOX}{Super dependent S-box} & 6 + open problem\\
    \hline
  6 &  \hyperlink{pr-entanglement}{Quantum entanglement} &  6 + open problem\\
    \hline
  7 &\hyperlink{pr-numberspoints}{Numbers and points}  & 4\\
    \hline
  8 & \hyperlink{pr-bob}{Bob's symbol} & 8\\
    \hline
 9 & \hyperlink{pr-cryptolocks}{Crypto locks}  & 4 + open problem \\
    \hline
  10 & \hyperlink{pr-publickeys}{Public keys for e-coins}  & open problem \\
     \hline
  11 &  \hyperlink{pr-long}{A long-awaited event} & 4\\
  
   \hline

\end{tabular}
\medskip

\caption{{Problems of the second round} }
\label{Probl-Second}

\end{table}

\section{Problems and their solutions}\label{problems}

In this section, we formulate all the problems of 2022 year Olympiad and present their detailed solutions, in some particular cases we also pay attention to solutions proposed by the participants.

\subsection{Problem ``Numbers and points''}

\subsubsection{Formulation}
\hypertarget{pr-numberspoints}{}

Decrypt the message in Fig.~\ref{NaP_fig1}.

\begin{figure}[ht]
\centering
\includegraphics[width=0.7\textwidth]{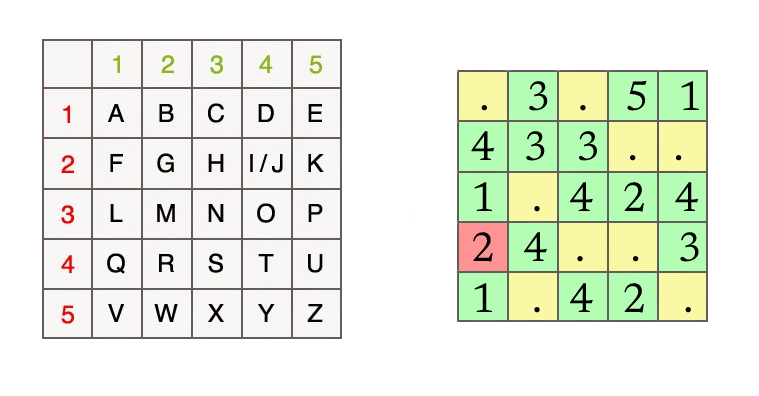}

\refstepcounter{figureCOUNTER}
{\small Fig.~\arabic{figureCOUNTER}. The illustration for the problem ``Numbers and points''}
\label{NaP_fig1}
\end{figure}

\subsubsection{Solution}
There is a board made up of numbers and dots on the right half of Fig.~\ref{NaP_fig1}. One cell is highlighted in red. The path along which the sensible plaintext is encrypted begins with it (Fig.~\ref{NaP_fig2}). The ciphertext has a <<number~--~number~--~dot>>  pattern. The ciphertext is the following:
\begin{center}
	$21$ . $42$ . $24$ . $15$ . $33$ . $14$ . 
\end{center}

The table in the left half of Fig.~\ref{NaP_fig1} refers to the Polybius square. Each letter is represented by its coordinates in the grid. Comparing the numbers from the ciphertext with the coordinates of the letters in the Polybius square, we get:
\begin{center}
	F . R. (I/J) . E . N . D .
\end{center}
Picking I from (I/J), we get the sensible plaintext \textbf{FRIEND}.

\begin{figure}[!h]
	\centering
		\includegraphics[scale=0.45]{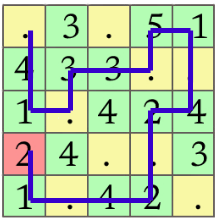} 

  \refstepcounter{figureCOUNTER}
{\small Fig.~\arabic{figureCOUNTER}. The path along which the sensible plaintext is encrypted}
  \label{NaP_fig2} 
\end{figure}

\bigskip
The problem looked simple but there was only one complete solution proposed by the team of Robin Jadoul (Belgium), Esrever Yu (Taiwan) and Jack Pope (United Kingdom).

\subsection{Problem ``Wallets''}

\subsubsection{Formulation}
\hypertarget{pr-wallet}{}

Bob has a wallet with 2022 NSUcoins. He decided to open a lot of new wallets and spread his NSUcoins among them. The platform that operates his wallets can distribute content of any wallet between 2 newly generated ones, charging 1 NSUcoin commission and removing the initial wallet.

He created a lot of new wallets, but suddenly noticed that all of his wallets contain exactly 8 NSUcoins each. Bob called the platform and told that there might be a mistake. How did he notice that?

\subsubsection{Solution}
Suppose that there were $n$ such operations, so we had $n+1$ wallets. Since 1 NSUcoin is charged for each operation, the total comission is equal to $n$. Therefore, we have $2022 - n=8(n+1)$ and $2014=9n$, but that is impossible since $n$ is a natural number. The most accurate and detailed solution was sent by Egor Desyatkov (Russia).

\subsection{Problem ``A long-awaited event''}

\subsubsection{Formulation}
\hypertarget{pr-long}{}

Bob received from Alice the secret message
\begin{center}
	{\tt L78V8LC7GBEYEE}
\end{center}
informing him about some important event.

It is known that Alice used an alphabet with $37$ characters from {\tt A} to {\tt Z}, from {\tt 0} to {\tt 9} and a space. Each of the letters is encoded as follows:
\begin{center}
	\begin{tabular}{|*{19}{c|}}
		\hline
		{\tt A} & {\tt B} & {\tt C} & {\tt D} & {\tt E} & {\tt F} & {\tt G} & {\tt H} & {\tt I} & {\tt J} & {\tt K} & {\tt L} & {\tt M} & {\tt N} & {\tt O} & {\tt P} & {\tt Q} & {\tt R} & {\tt S}  \\\hline
		$0$ & $1$ & $2$ & $3$ & $4$ & $5$ & $6$ & $7$ & $8$ & $9$ & $10$ & $11$ & $12$ & $13$ & $14$ & $15$ & $16$ & $17$ & $18$  \\\hline\hline
		{\tt T} & {\tt U} & {\tt V} & {\tt W} & {\tt X} & {\tt Y} & {\tt Z} & {\tt 0} & {\tt 1} & {\tt 2} & {\tt 3} & {\tt 4} & {\tt 5} & {\tt 6} & {\tt 7} & {\tt 8} & {\tt 9} &  \multicolumn{2}{c|}{SPACE}  \\\hline
		$19$ & $20$ & $21$ & $22$ & $23$ & $24$ & $25$ & $26$ & $27$ & $28$ & $29$ & $30$ & $31$ & $32$ & $33$ & $34$ & $35$ & \multicolumn{2}{c|}{$36$} \\\hline
	\end{tabular}
\end{center}

\smallskip

 For the encryption, Alice used a function $f$ such that $f(x)=ax^2+bx+c \pmod{37}$ for some integers $a,b,c$ and $f$ satisfies the property
\begin{equation*}
	f(x-y)-2f(x)f(y)+f(1+xy)=1 \pmod{37}\ \ \text{for any integers }x,y.
\end{equation*}

Decrypt the message that Bob has received.

\subsubsection{Solution}

Let $y=0$:
\begin{equation*}
	f(x)-2f(x)f(0)+f(1)=1 \pmod {37},
\end{equation*}
\begin{equation*}
	f(x)(1-2f(0))=1-f(1) \pmod {37}.
\end{equation*}
Since $f$ is not a constant function, we have that both sides of the equation above are zeros, so $f(0)=19\pmod {37}$ and $f(1)=1\pmod {37}$. From this we obtain that~$c=19$. Let $y=-1$:
\begin{equation*}
	f(1+x)+f(1-x)=1+2f(x)f(-1) \pmod {37}.
\end{equation*}
By replacing $x \mapsto (-x)$ we get
\begin{equation*}
	f(1-x)+f(1+x)=1+2f(-x)f(-1) \pmod {37}.
\end{equation*}
Left sides of the last two expressions are equal, therefore $f(x)=f(-x)\pmod {37}$ that is~$f$ is even function, provided $f(-1)\ne0 \pmod {37}$. We can check the last condition buy putting $x=0$, $y=1$ to the initial relation on~$f$, that yields $f(-1)=1\ne0 \pmod {37}$. Therefore, $f(x)=f(-x)\pmod {37}$ for any integer~$x$, hence~$b=0$.

From $f(1)=1\pmod {37}$ we reveal the value of the coeffecient~$a$ that is equal to $19$. Thus, we have~$f(x)=19\big(x^2+1\big)\pmod {37}$, then for recovering of the plaintext we use the inverse expression $x=\pm\sqrt{2f(x)+36}\pmod {37}$ and for every symbol of the ciphertext we choose the appropriate variant of the corresponding symbol of the plaintext:
\begin{equation*}
	\text{L78V8LC7GBEYEE} \, \hookrightarrow \, \text{NSUCRYPTO 2022}.
\end{equation*}

The only correct solution was sent by William Zhang (United Kingdom).

\subsection{Problem ``Hidden primes''}

\subsubsection{Formulation}
\hypertarget{pr-hiddenprimes}{}

The Olympiad team rented an office at the Business Center, 1-342 room, on 1691th street for NSUCRYPTO-2022 competition for 0 nsucoins (good deal!). Mary from the team wanted to create a task for the competition and she needed to pick up three numbers for this task. She used to find an inspiration in numbers around her and various equations with them. After some procedure she found three prime numbers! It is interesting that when Mary added the smallest number to the largest one and divided the sum by the third number, the result was also the prime number. 

Could you guess these numbers she found?

\begin{figure}[ht]
\begin{center}
\includegraphics[width=0.6\textwidth]{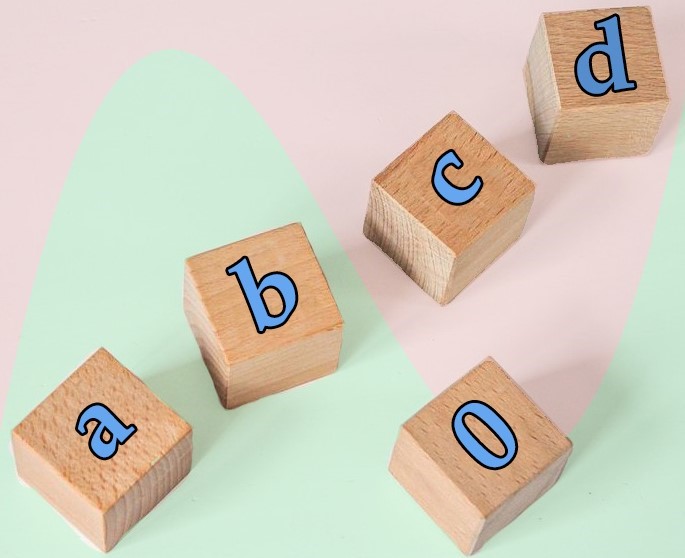}

\refstepcounter{figureCOUNTER}
{\small Fig.~\arabic{figureCOUNTER}. The illustration for the problem ``Hidden primes''}

\label{fig:cubics}
\end{center}
\end{figure}
\subsubsection{Solution}

 We may assume from the problem statement that Mary used some numbers around her and some equations with them in order to find these three numbers. We may also get from the description that she used only one procedure to find these hidden numbers. 
 
 So, all three numbers are connected by some procedure and the numbers around Mary are used, from phrase ``various equations'' we can assume that there exists some equation with these numbers as coefficients. There were 5 numbers around Mary: 1, -342, 1691, -2022 and 0.

In addition, analyzing the picture (see Fig.~\ref{fig:cubics}), you can see the curve, cubes with 4 letters: a, b, c, d and the cube with 0. The curve resembles a graph of a cubic function and the letters on the cubes look like coefficients of a cubic function. The cube with 0 gives a hint for the use of a cubic equation.

Let us substitute the numbers from the problem statement into the cubic equation. Solving the equation $x^3 - 342 x^2 + 1691 x - 2022 = 0$ we find the roots 2, 3, 337. All three numbers are prime and satisfy the condition from the statement: $(2 + 337)/3 = 113$, where 113 is also a prime number.

\bigskip

Best solutions were proposed independently by Konstantin Romanov (Russia), Vasiliy Kadykov (Russia) and Sergey Zabolotskiy (Russia).
\bigskip

\subsection{Problem ``Face-to-face''}

\subsubsection{Formulation}
\hypertarget{pr-facetoface}{}

Alice picked a new pin code (4 pairwise distinct digits from $\{1,2,\ldots,9\}$) for her credit card such that all digits have the same parity and are arranged in increasing order.
Bob and Charlie wanted to guess her pin code. Alice said that she can give each of them a hint but face-to-face only.

Bob alone came to Alice and she told him that the sum of her pin code digits is equal to the number of light bulbs in the living room chandelier. Bob answered that there is still no enough information for him to guess the code, and left.
After that, Charlie alone came to Alice and she told him that if we find the product of all pin code digits and then sum up digits of those product, this result number would be equal to the amount of books on the shelf. Charlie also answered that there is still no enough information for him to guess the code, and left.

Unfortunately, Eve was eavesdropping in the next apartment and, after Charlie had left, she immediately found out Alice pin code despite that she had never seen those chandelier and bookshelf. Could you find the pin code too?

\subsubsection{Solution}
Let $P$ be the pin code. Since all the digits of $P$ have the same parity and are arranged in increasing order, we have only six options:

\medskip

\begin{center}
		{\footnotesize
			\begin{tabular}{|c|c|c|c|}
				\hline
    Pin code $P$ & The sum of digits & The product of digits &  The sum of product digits\\
    \hline
				1357 & 16 & 105 & 6\\
				\hline
				1359 &  18 & 135 & 9\\
				\hline
				1379 & 20 & 189 & 18\\
				\hline
    1579 & 22 & 315  & 9\\
    \hline
     2468 & 20 & 384 & 15\\
				\hline
    3579 & 24 & 945 & 18\\
				\hline
   
			\end{tabular}
		}
	\end{center}
\medskip

Since Bob could not guess the code, the sum of digits must allow at least two options for the code, so, we have that $P \in \{1379, 2468\}$. Since Charlie could not guess the code either, we have the same problem for the sum of product digits and it follows that $P \in \{1359, 1579, 1379, 3579\}$. Therefore, the pin code is equal to $1379$.

Best solutions of this problem were sent by Henning Seidler (Germany), Himanshu Sheoran (India) and Phuong Hoa Nguyen (France).

\subsection{Problem ``Crypto locks''}

\subsubsection{Formulation}
\hypertarget{pr-cryptolocks}{}

Alice and Bob are wondering about the creation of a new version for the Shamir three-pass protocol. They have several ideas about it.

The Shamir three-pass protocol was developed more than 40 years ago. Recall it. Let~$p$ be a big prime number. Let Alice take two secret numbers $c_A$ and $d_A$ such that $c_A d_A = 1 \bmod{(p-1)}$. Bob takes numbers $c_B$ and $d_B$ with the same property. If Alice wants to send a secret message $m$ to Bob, where $m$ is an integer number $1< m < p-1$, then she calculates $x_1 = m^{c_A} \bmod p$ and sends it to Bob. Then Bob computes $x_2 = x_1^{c_B} \bmod p$ and forwards it back to Alice. On the third step, Alice founds $x_3 = x_2^{d_A}\bmod p$ and sends it to Bob. Finally, Bob recovers $m$ as $x_3^{d_B}\bmod p$ according to Fermat's Little~theorem.

It is possible to think about action of $c_A$ and $d_A$ over the message as about locking and unlocking, see Fig.~\ref{fig:locks}.

\begin{figure}[ht]
\begin{center}
\includegraphics[width=0.7\textwidth]{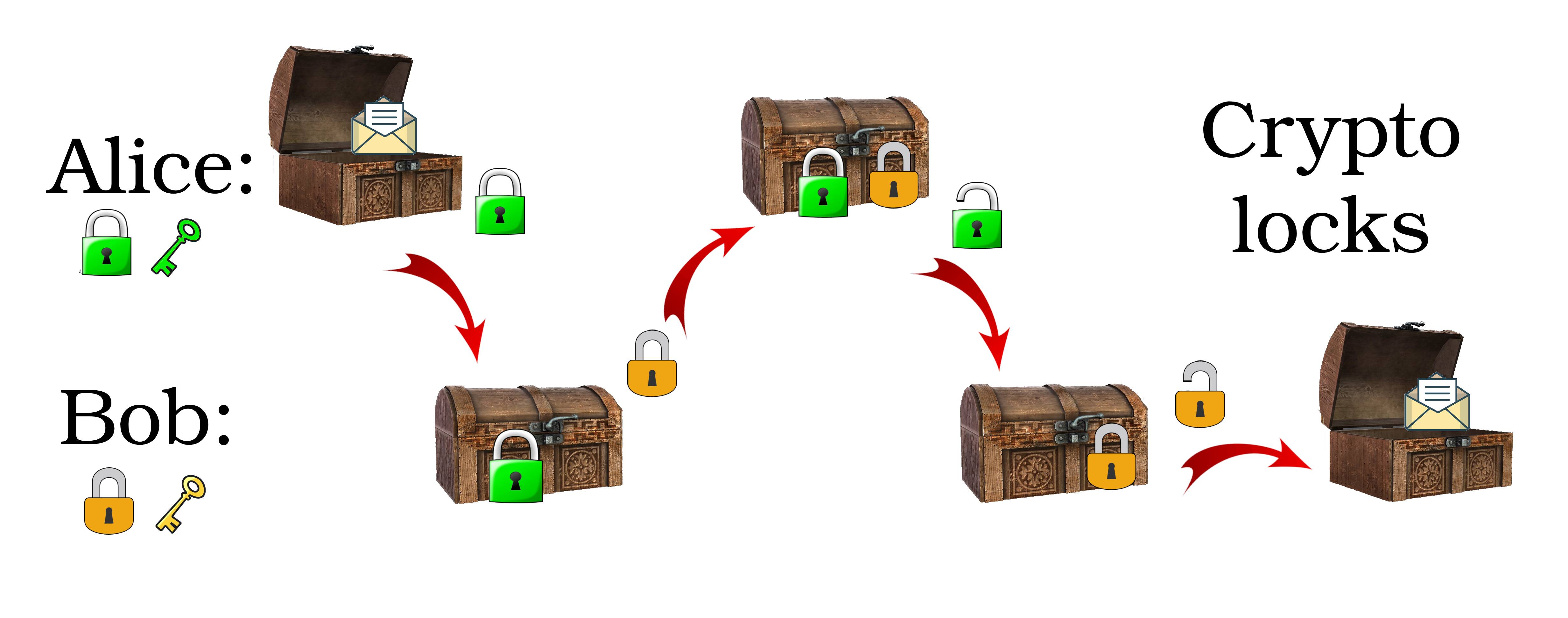}

\refstepcounter{figureCOUNTER}
{\small Fig.~\arabic{figureCOUNTER}. The illustration for the problem ``Crypto locks''}

\label{fig:locks}
\end{center}
\end{figure}

Alice and Bob decided to change the scheme by using symmetric encryption and decryption procedures instead of locking and unlocking with $c_A$, $c_B$, $d_A$ and $d_B$.

\begin{itemize}
\item[{\bf Q1}] Propose some simple symmetric ciphers that would be possible to use in such scheme. What properties for them are required? Should Alice and Bob use the same cipher (with different own keys) or not?

\item[{\bf Q2}] \underline{\color{blue}{\bf Problem for a special prize!}} Could you find such symmetric ciphers that make the modified scheme to be secure as before? Please, give your reasons and proofs.
\end{itemize}

\subsubsection{Solution}

 \textbf{Q1}. Assume that Alice and Bob use functions Enc$_A$, Dec$_A$ and Enc$_B$, Dec$_B$ for encryption and decryption, respectively. Suppose that Alice wants to send the message $m$, then the three-pass protocol will look as follows:
 
 $\bullet$ Alice calculates Enc$_A(m,k_A)$, where $k_A$ is her secret key, and sends it to Bob;
 
 $\bullet$ Bob computes Enc$_B(\text{Enc}_A(m,k_A),k_B)$, where $k_B$ is his secret key, and forwards it to Alice;
 
 $\bullet$ Finally, Alice computes Dec$_A(\text{Enc}_B(\text{Enc}_A(m,k_A),k_B),k_A)$ and sends it to Bob;
 
 In order for Bob to recover $m$ the following property must hold 
 $$
 \text{Dec}_B(\text{Dec}_A(\text{Enc}_B(\text{Enc}_A(m,k_A),k_B),k_A),k_B)=m.
 $$
 
 The most common approach was to use encryption functions that commute with each other. In that case, if Alice wants to send a secret message $m$ to Bob, then she calculates $x = m\circ k_A$ and sends it to Bob. Then Bob computes $x_2 = x\circ k_B$ and forwards it back to Alice. On the third step, Alice finds $x_3 = x_2\circ k_A^{-1}$ and sends it to Bob. Finally, the commutative property of operation $\circ$ allows Bob to recover $m$ as $x_3\circ k_B^{-1}$.
 
 \begin{remark}\label{refCryptoLocks}
 	Note that if Eve can intercept all three messages, then she can obtain $m$ if she could compute $x_2^{-1}$, since $x\circ x_3\circ x_2^{-1}=m$. As a result, all schemes that use ciphers with only XOR operation (the most common suggestion by the participants) have this weakness.
 \end{remark}

Regarding \textbf{Q2}, one interesting idea found by a few participants is to use product of matrices for encryption and decryption, with the additional condition that the matrix $M$ associated with the message $m$ is singular. That additional condition appears as a countermeasure against the attack described in Remark~\ref{refCryptoLocks}. However, such schemes require additional security analysis.

Another interesting idea suggested by the team of Himanshu Sheoran, Gyumin Roh and Yo Iida (India, South Korea, Japan) was to base the scheme on permutations that commute with each other. Note that a three-pass cryptographic protocol with a similar idea was presented in~\cite{PermThreePass}.

\subsection{Problem ``Matrix and reduction''}
\subsubsection{Formulation}
\hypertarget{pr-matrixreduction}{}

Alice used an alphabet with $30$ characters from {\tt A} to {\tt Z} and {\tt 0}, {\tt 1}, <<{\tt ,}>>, <<{\tt !}>>. Each of the letters is encoded as follows:

\begin{center}
	\begin{tabular}{|*{15}{c|}}
		\hline
		{\tt A} & {\tt B} & {\tt C} & {\tt D} & {\tt E} & {\tt F} & {\tt G} & {\tt H} & {\tt I} & {\tt J} & {\tt K} & {\tt L} & {\tt M} & {\tt N} & {\tt O} \\\hline
		$0$ & $1$ & $2$ & $3$ & $4$ & $5$ & $6$ & $7$ & $8$ & $9$ & $10$ & $11$ & $12$ & $13$ & $14$\\\hline\hline
		{\tt P} & {\tt Q} & {\tt R} & {\tt S} & {\tt T} & {\tt U} & {\tt V} & {\tt W} & {\tt X} & {\tt Y} & {\tt Z} & {\tt 0} & {\tt 1} & {\tt ,} & {\tt !}  \\\hline
		$15$ & $16$ & $17$ & $18$ & $19$ & $20$ & $21$ & $22$ & $23$ & $24$ & $25$ & $26$ & $27$ & $28$ & $29$\\\hline
	\end{tabular}
\end{center}

\smallskip
\noindent\textbf{Encryption.} The plaintext is divided into consequent subwords of length~$4$ that are encrypted independently via the same encryption $(2\times2)$-matrix~$F$ with elements from~$\mathbb{Z}_{30}$. For example, let the $j$-th subword be {\tt WORD} and the encryption matrix~$F$ be equal to
\begin{equation*}
F=\begin{pmatrix}
	11 & 9 \\
	11 & 10 \\
\end{pmatrix}.
\end{equation*}
The matrix that corresponds to {\tt WORD} is denoted by~$P_j$ and the matrix that corresponds to the result of the encryption of~{\tt WORD} is~$C_j$ and calculated as follows:
\begin{equation*}
C_j=F\cdot P_j=
\begin{pmatrix}
	11 & 9 \\
	11 & 10 \\
\end{pmatrix}
\cdot
\begin{pmatrix}
	22 & 17 \\
	14 & 3 \\
\end{pmatrix}
=
\begin{pmatrix}
	8 & 4 \\
	22 & 7 \\
\end{pmatrix}
\pmod{30},
\end{equation*}
that is the $j$-th subword of the ciphertext is {\tt IWEH}.

\bigskip

Eve has intercepted a ciphertext that was transmitted from Alice to Bob:
\begin{equation*}
	{\tt CYPHXWQE!WNKHZ0Z}
\end{equation*}
Also, she knows that the third subword of the plaintext is {\tt FORW}. Will Eve be able to restore the original message?

\subsubsection{Solution}
The third word of the plaintext is $\mathtt{FORW}$: 
\begin{equation*}
	P=\mathtt{FORW}=
	\begin{pmatrix}
		5 & 17 \\
		14 & 22 \\
	\end{pmatrix}\pmod {30}.
\end{equation*} 
The ciphertext corresponding to it:
\begin{equation*}
	C=\mathtt{!WNK}=
	\begin{pmatrix}
		29 & 13 \\
		22 & 10 \\
	\end{pmatrix}\pmod {30}.
\end{equation*} 
Since $C_3=F\cdot P_3$, where $F$ is the encryption matrix, the matrix for the decryption could have the following form:
\begin{equation*}
	D=P_3 \cdot C^{-1}_3.
\end{equation*} 
But $\mathrm{det}\big(C_3\big)=4\pmod {30}$ and $\text{gcd}(4,30)\neq 1$, that is such matrix does not exist modulo~$30$. 

So, we will consider following calculations by reduction modulo $15$. Let $\overline{P_3}=P_3\pmod {15}$, $\overline{C_3}=C\pmod {15}$ and $\overline{F}=F\pmod {15}$. We have
\begin{equation*}
	\overline{F}^{-1}=\overline{P_3}\cdot\big(\overline{C_3}\big)^{-1}=
	\begin{pmatrix}
		9 & 2 \\
		4 & 9 \\
	\end{pmatrix}\pmod {15},
\end{equation*} 
consequently,
\begin{equation*}
D=
\begin{pmatrix}
	9 & 2 \\
	4 & 9 \\
\end{pmatrix}
+15 F_0\pmod {30},
\end{equation*} 
where $F_0$ is $2\times2$ binary matrix. 

We have $D\cdot C_3 =P_3$ or
\begin{equation*}
	\overline{F}^{-1}\cdot 
	\begin{pmatrix}
		29 & 13 \\
		22 & 10 \\
	\end{pmatrix}
	+15\cdot F_0 \cdot 
	\begin{pmatrix}
		29 & 13 \\
		22 & 10 \\
	\end{pmatrix}
	=
	\begin{pmatrix}
		5 & 17 \\
		14 & 22 \\
	\end{pmatrix}\pmod {30}.
\end{equation*} 
Finally, we obtain
\begin{equation*}
	\begin{pmatrix}
		5 & 17 \\
		14 & 22 \\
	\end{pmatrix}
	=F_0 \cdot 
	\begin{pmatrix}
		15 & 15 \\
		0 & 0 \\
	\end{pmatrix}
	=
	\begin{pmatrix}
		5 & 17 \\
		14 & 22 \\
	\end{pmatrix}\pmod {30}.
\end{equation*} 
If we set $F_0=\begin{pmatrix}
	a & b \\
	c & d \\
\end{pmatrix}$ then it is clear that only the values $a=c=0$ and only the values $b=1$, $d=0$ give us the answer {\tt GOODLUCKFORWIN!!}. 

Best solutions for this problem were sent by Pieter Senden (Belgium), 
and by Sergey Zabolotskiy (Russia). 

\subsection{Problem ``Reversing a gate''}

\subsubsection{Formulation}
\hypertarget{pr-reversingagate}{}

	Daniel continues to study quantum circuits. A controlled NOT (CNOT) gate is the most complex quantum gate from the universal set of gates required for quantum computation. This gate acts on two qubits and makes the following transformation:
	\begin{equation*}
		\ket{00}\rightarrow\ket{00}, \quad \ket{01}\rightarrow\ket{01}, \quad \ket{10}\rightarrow\ket{11}, \quad \ket{11}\rightarrow\ket{10}.
	\end{equation*}
	This gate is clearly asymmetric. The first qubit is considered as control one, and the second is as a target one. CNOT is described by the following quantum circuit ($x,y\in\mathbb{F}_2$):
	\begin{center}
		\begin{quantikz}[ampersand replacement=\&]
			\lstick{$\ket{x}$} \& \ctrl{1} \& \qw \rstick{$\ket{x}$} \\
			\lstick{$\ket{y}$} \& \targ{} \& \qw \rstick{$\ket{y\oplus x}$}
		\end{quantikz}
	\end{center}
	{\bf The problem.} Help Daniel to design a circuit in a special way that reverses CNOT~gate:
		\begin{center}
		\begin{quantikz}[ampersand replacement=\&]
			\lstick{$\ket{x}$} \& \targ{} \& \qw \rstick{$\ket{x\oplus y}$} \\
			\lstick{$\ket{y}$} \& \ctrl{-1} \& \qw \rstick{$\ket{y}$}
		\end{quantikz}
	\end{center}
	It makes the following procedure: $\ket{00}\rightarrow\ket{00},~\ket{01}\rightarrow\ket{11},~\ket{10}\rightarrow\ket{10},~\ket{11}\rightarrow\ket{01}$.	
	To do this you should modify the original CNOT gate without re-ordering the qubits but via adding some single-qubit gates instead from the following ones:
	\savebox{\boxX}{\footnotesize \begin{quantikz}[ampersand replacement=\&]
			\lstick{$\ket{x}$} \& \gate{X} \& \qw \rstick{$\ket{x\oplus1}$}
	\end{quantikz} }
	\savebox{\boxZ}{\footnotesize \begin{quantikz}[ampersand replacement=\&]
			\lstick{$\ket{x}$} \& \gate{Z} \& \qw \rstick{$(-1)^x\ket{x}$}
	\end{quantikz} }
	\savebox{\boxH}{\footnotesize \begin{quantikz}[ampersand replacement=\&]
			\lstick{$\ket{x}$} \& \gate{H} \& \qw \rstick{$\frac{\ket{0}+(-1)^x\ket{1}}{\sqrt{2}}$}
	\end{quantikz}}
	\begin{center}
		{\footnotesize
			\begin{tabular}{|p{2.4cm}|l|p{5cm}|}
				\hline
				Pauli-X gate & \usebox\boxX & acts on a single qubit in the state $\ket{x}$, $x\in\{0,1\}$\\
				\hline
				Pauli-Z gate & \usebox\boxZ & acts on a single qubit in the state $\ket{x}$, $x\in\{0,1\}$\\
				\hline
				Hadamard gate & \usebox\boxH & acts on a single qubit in the state $\ket{x}$, $x\in\{0,1\}$\\
				\hline
			\end{tabular}
		}
	\end{center}
	
	\bigskip
	{\small
		\noindent{\bf Remark.} Let us briefly formulate the key points of quantum circuits.
		A qubit is a two-level quantum mechanical system whose state $\ket{\psi}$ is the superposition of basis quantum states $\ket{0}$ and $\ket{1}$. The superposition is written as $\ket{\psi}=\alpha_0\ket{0}+\alpha_1\ket{1}$, where $\alpha_0$ and $\alpha_1$ are complex numbers, called amplitudes, that possess $\left|\alpha_0\right|^2+\left|\alpha_1\right|^2=1$. The amplitudes $\alpha_0$ and $\alpha_1$ have the following physical meaning: after the measurement of a qubit which has the state $\ket{\psi}$, it will be observed in the state $\ket{0}$ with probability $\left|\alpha_0\right|^2$ and in the state $\ket{1}$ with probability $\left|\alpha_1\right|^2$.
			In order to operate with multi-qubit systems, we consider the bilinear operation $\otimes:\ket{x},\ket{y}\rightarrow\ket{x}\otimes\ket{y}$ on $x,y\in\{0,1\}$ which is defined on pairs $\ket{x},\ket{y}$, and by bilinearity is expanded on the space of all linear combinations of $\ket{0}$ and $\ket{1}$. When we have two qubits in states $\ket{\psi}$ and $\ket{\varphi}$ correspondingly, the state of the whole system of these two qubits is
		$
		\ket{\psi}\otimes\ket{\varphi}.
		$
		In general, for two qubits we have
		$
		 \ket{\psi}=\alpha_{00}{\ket{0}\otimes\ket{0}}+\alpha_{01}\ket{0}\otimes\ket{1}+\alpha_{10}\ket{1}\otimes\ket{0}+\alpha_{11}\ket{1}\otimes\ket{1}.
		$
		The physical meaning of complex numbers $\alpha_{ij}$ is the same as for one qubit, so we have the essential restriction $|\alpha_{00}|^2+|\alpha_{01}|^2+|\alpha_{10}|^2+|\alpha_{11}|^2=1$. We use more brief notation $\ket{a}\otimes\ket{b}\equiv\ket{ab}$.
		In order to verify your circuits, you can use different quantum circuit simulators, for example, see \cite{quirk}.
	}

\subsubsection{Solution}

The desired circuit has the following form for any~$x,y\in\mathbb{F}_2$:
\begin{center}
	\begin{quantikz}[slice all,remove end slices=0,slice titles=\ket{\psi_{\col}},slice style=blue,ampersand replacement=\&]
		\lstick{\ket{x}} \& \gate{H} \& \ctrl{1} \& \gate{H} \& \qw \rstick{\ket{x\oplus y}} \\
		\lstick{\ket{y}} \& \gate{H} \& \targ{} \& \gate{H} \& \qw \rstick{\ket{y}}
	\end{quantikz}
\end{center}

Indeed, with initial state~$\ket{x}\ket{y}$ we have
\begin{align*}
	\ket{\psi_1}&=
	\left(\frac{\ket{0}+(-1)^x\ket{1}}{\sqrt{2}}\right)\left(\frac{\ket{0}+(-1)^y\ket{1}}{\sqrt{2}}\right)\\
	&=\frac{\ket{00}+(-1)^y\ket{01}+(-1)^x\ket{10}+(-1)^{x\oplus y}\ket{11}}{2},\\
	\ket{\psi_2}&=
	\frac{\ket{00}+(-1)^y\ket{01}+(-1)^x\ket{11}+(-1)^{x\oplus y}\ket{10}}{2}\\
	&=\left(\frac{\ket{0}+(-1)^{x\oplus y}\ket{1}}{\sqrt{2}}\right)\left(\frac{\ket{0}+(-1)^y\ket{1}}{\sqrt{2}}\right),\\
	\ket{\psi_3}&=\ket{x\oplus y}\ket{y}.
\end{align*}

Best solutions were sent by Daniel Popescu (Romania), by Yo Iida (Japan) and by David Marton (Hungary).

\subsection{Problem ``Bob's symbol''}

\subsubsection{Formulation}
\hypertarget{pr-bob}{}

Bob learned the Goldwasser--Micali cryptosystem at university. Now, he is thinking about  functions over finite fields that are similar to Jacobi symbol.

He chose a function $B_n: \mathbb{F}_{2^n} \to \mathbb{F}_2$ (Bob's symbol) defined as follows for any $a \in \mathbb{F}_{2^n}$:
$$
	B_n(a) = \begin{cases}
		1, & \text{if } a = x^2 + x \text{ for some } x \in \mathbb{F}_{2^n},\\
		0, & \text{otherwise}.
	\end{cases}
$$

Bob knows that finite fields may have some subfields. Indeed, it is well known that $\mathbb{F}_{2^k}$ is a subfield of $\mathbb{F}_{2^n}$ if and only if $k \mid n$. Bob wants to exclude the elements of subfields. In other words, he considers the restriction of $B_n$ to the set
$$
	\widehat{\mathbb{F}}_{2^n} = \mathbb{F}_{2^n} \setminus \bigcup_{k \mid n,\, k \neq n} \mathbb{F}_{2^k}.
$$
Here, by $\mathbb{F}_{2^n} \setminus \mathbb{F}_{2^k}$ we mean the removal from $\mathbb{F}_{2^n}$ the elements forming the field of order~$2^k$.

Finally, Bob is interested in the sets $$B_n^0 = \{ y \in \widehat{\mathbb{F}}_{2^n} : B_n(y) = 0\} \text{~~~and~~~} B_n^1 = \{ y \in \widehat{\mathbb{F}}_{2^n} : B_n(y) = 1\}.$$

\begin{enumerate}
	\item[{\bf Q1}] Help Bob to find ${|B_n^0|}/{|B_n^1|}$ if $n$ is odd.
	\item[{\bf Q2}] Help Bob to find $|B_n^0|$ and $|B_n^1|$ for an arbitrary $n$.
\end{enumerate}

\subsubsection{Solution}

Let us define
\begin{align*}
	\fdn{2^n} &= \{ x \in \fd{2^n} : 
	\bobsym{x}{n} = 0 \}, \text{ i.e. } \bobset{n} = \fdo{2^n} \cap \fdn{2^n}.
\end{align*}
First we prove the following lemma.
\begin{lemma}\label{lemma::p2}
    Let $k \mid n$. Then 
    $$
        |\fd{2^k} \cap \fdn{2^n}|  = 
        \begin{cases}
            \frac{1}{2}|\fd{2^k}|, & \text{if } n/k\text{ is odd},\\
            0, & \text{otherwise}.
        \end{cases}
    $$
\end{lemma}
\begin{proof} 
	Let us consider the function $G(x) = x^2 + x = x(x + 1)$, where $x \in \fd{2^k}$. First, $G(x) = G(x + 1)$. Secondly, $x^2 + x + a$, $a \in \fd{2^k}$, has at most $2$ roots. It means that $G$ is a two-to-one function. Therefore, there are exactly $2^{k - 1}$ distinct $a$ such that $x^2 + x \neq a$ for any $x \in \fd{2^k}$.	
	
	Next, for any such $a$ the polynomial $x^2 + x + a$ is irreducible over $\fd{2^k}$. It means that it has a root $q$ in the quadratic extension $\fd{2^{2k}}$ of $\fd{2^k}$, i.e. $a = q^2 + q$. If $n/k$ is even, $\fd{2^{2k}}$ is a subfield of $\fd{2^n}$, i.e. $q \in \fd{2^n}$. Thus, $|\fd{2^k} \cap \fdn{2^n}| = 0$. If $n/k$ is odd, then $\fd{2^{2k}}$ is not a subfield of $\fd{2^n}$. Moreover, $\fd{2^{2k}} \cap \fd{2^n} = \fd{2^{k}}$. It means that any root $q$ does not belong to $\fd{2^n}$, i.e. $|\fd{2^k} \cap \fdn{2^n}| = 2^{k - 1}$.    	    
\end{proof}

Now we are ready to answer the questions. Let $n = m 2^t$, where $m$ is odd. We define 
\begin{equation*}
	f_t(d) = |\fdo{2^{d2^t}} \cap \fdn{2^n}| \text{ and } g_t(d) = \frac{1}{2}2^{d 2^t},
\end{equation*}
where $d \mid m$. This means that $|\bobset{n}| = |\bobset{m 2^t}| = f_t(m)$. At the same time, the definition of $\fdo{2^n}$ gives us that
$$
	\sum_{d \mid n}	|\fdo{2^{d}} \cap \fdn{2^n}| = |\fd{2^{n}} \cap \fdn{2^n}|.
$$
According to Lemma~\ref{lemma::p2} and the denotations above,
\begin{equation*}
	\begin{gathered}
		\sum_{d \mid n}	|\fdo{2^{d}} \cap \fdn{2^n}| = \sum_{d \mid m}|\fdo{2^{d2^t}} \cap \fdn{2^n}| = \sum_{d \mid m} f_t(d) \text{ and }\\
		|\fd{2^{n}} \cap \fdn{2^n}| = |\fd{2^{m 2^t}} \cap \fdn{2^n}| = \frac{1}{2} |\fd{2^{m 2^t}}| =  g_t(m).
	\end{gathered}
\end{equation*}
Hence,
$$
	g_t(m) = \sum_{d \mid m} f_t(d) \text{ holds for any integers } m \geq 1 \text{ and } t \geq 0.
$$
According to the Möbius inversion formula,
$$
	f_t(m) = \sum_{d \mid m} \mu(d) g_t(m/d) = \frac{1}{2} \sum_{d \mid m} \mu(d) 2^{(m / d) 2^t}.
$$
Recall that $\mu(d) = 0$ if $d$ is not square-free (there is an integer $u \geq 2$ such that $u^2 \mid d$); otherwise, it is equal to $1$ ($-1$ resp.) if $d$ has an even (odd resp.) number of prime factors. As a result,
$$
	|\bobset{n}| =  \frac{1}{2} \sum_{d \mid m} \mu(d) 2^{n / d}.
$$

Also, $|\bobseta{n}| = |\fdo{2^n}| - |\bobset{n}|$. We need only to note that 
    $$
        |\fdo{2^n}| = \sum_{d \mid n} \mu(d) 2^{n/d}.
    $$
This can be easily proven just using
$$
	2^n = |\fd{2^n}| = \sum_{d \mid n} |\fdo{2^d}|
$$
together with the Möbius inversion formula. 
Finally, we can see that $|\bobset{n}| = |\bobseta{n}| = \frac{1}{2}|\fdo{2^n}|$ for odd $n$, which means that the answer for Q1 is $1$. In fact, it directly follows from  Lemma~\ref{lemma::p2} and the definition of $\fdo{2^n}$. 

Many teams provided the correct answers in the second round using similar ideas: Himanshu Sheoran, Gyumin Roh, Yo Iida (India), Mikhail Kudinov, Denis Nabokov, Alexey Zelenetskiy (Russia), Stepan Davydov, Anastasiia Chichaeva, Kirill Tsaregorodtsev (Russia), Mikhail Borodin, Vitaly Kiryukhin, Andrey Rybkin (Russia), Kristina Geut, Sergey Titov, Dmitry Ananichev (Russia), Pham Minh, Dung Truong Viet (Vietnam) and Alexander Belov (Russia).

\subsection{Problem ``Public keys for e-coins''}

\subsubsection{Formulation}
\hypertarget{pr-publickeys}{}

Alice has $n$ electronic coins that she would like to spend via some public service $S$ (bank).
The service applies some asymmetric algorithm of encryption $E(,)$ and decryption $D(,)$ in its work. Namely, for the pair
of public and private keys $(PK, SK)$ and for any message $m$ it holds: if $c = E(m,PK)$, then $m=D(c,SK)$ and
visa versa: if $c'=E(m,SK)$, then $m=D(c',PK)$.

To spend her money, Alice generates a sequence of public and private key pairs
$(PK_1, SK_1), \ldots,$ $(PK_n, SK_n)$
and sends the sequence of public keys $PK_1, \ldots ,PK_n$ to the service $S$.
By this she authorizes the service $S$ to control her $n$ coins.

If Alice would like to spend a coin with number $i$ in the shop of Bob, she just gives
the secret key $SK_i$ to Bob and informs him about the number $i$. To get the coin with number~$i$, Bob sends to the service $S$ three
parameters: number $i$, some non secret message $m$, and its electronic signature $c' = E(m, SK_i)$. The service $S$ checks whether the signature $c'$ corresponds to the message $m$, i.e. does it hold the equality
$m = D(c', PK_i)$. If it is so, the service accepts the signature, gives the coin number $i$ to Bob and marks it as <<spent>>.

\underline{\color{blue}{\bf Problem for a special prize!}} Propose a {\it modification of this scheme} related to generation of public and private key pairs. Namely, is it possible for Alice not to send the sequence of public keys $PK_1, \ldots PK_n$ to the service $S$, but send only some initial information enough for generating all necessary public keys on the service's side?
Suppose that Alice sends to the service $S$ only some initial key $PK$ (denote it also as $PK_0$), some function $f$ and a set of parameters $T$ such that $PK_{i+1} = f(PK_i, T)$ for all $i\geqslant 0$.
Propose your variant of this function $f$ and the set $T$. Think also what asymmetric cryptosystem it is possible to use in such scheme.

\bigskip
\noindent{\bf Requirements to the solution.} 
Knowing $PK$, $f$ and $T$, it is impossible to find any private key $SK_i$, where $i=1, \ldots ,n$. It should be impossible to recover $SK_i$ even if the secret keys $SK_1,\ldots,SK_{i-1}$ are also known, or even if all other secret keys are known (more strong condition).

\subsubsection{Solution} The problem was solved by two teams and partially solved by three teams.

One of the best partial solutions was proposed by the team of Viet-Sang Nguyen, Nhat Linh Le Tan and Phuong Hoa Nguyen from France. It is based on principles of Elliptic-curves-cryptography and hash functions. The main idea is to consider $SK_i$ as the sequence of numbers related to each other with the help of HMAC-SHA256. Public keys can be easily generated by the server $S$. The main disadvantage of the scheme is described by the authors: server $S$ should keep the point $PK_0$ in secret, as well as Bob should do with $SK_i$. The problem is that if there is some data leakage, then all coins of Alice will be lost. So, the potential complicity of the server and Bob forms a crucial danger for Alice. 

An interesting idea was proposed by Himanshu Sheoran (India) Gyumin Roh (South Korea) and Yo Iida (Japan). It is based on the combination of two pairs of RSA keys. With one pair it is proposed to sign messages from Bob to the server, with another one Alice generates her private keys to give them to Bob. Solution was accepred as partial since the security of this scheme should be considered in more details. 

A very nice partial solution was proposed by Robin Jadoul (Belgium), Esrever Yu (Taiwan), Jack Pope (United Kingdom). The authors describe an identity-based signature scheme with message recovery based on the RSA hardness assumption. The main idea is to generate public and private keys from the corresponding master keys by application of cryptographic hash functions (four functions are used).

An original attempt to solve the problem was proposed by Alexander Bakharev, Rinchin Zapanov and Denis Bykov (Russia). They applied RSA-like technique and considered private keys as $SK_i = PK_i^{-1} \mod \phi(n)$, where $n=pq$ and prime numbers $p$, $q$ are known to Alice only, as well as $\phi(n)$. Public keys are formed as the consecutive prime numbers: $PK_{i+1}$ is the next prime number after $PK_i$. But the security of this scheme is still under the question since public keys are too connected; it should be analyzed. 

We have accepted two complete solutions.

One of them was proposed by the team of G.Teseleanu, P.Cotan, L.Constantin-Sebastian from the Institute of Mathematics of the Romanian Academy. On the first round the partial solution was proposed by G.Teseleanu. RSA-like technique is applied in the solution. Private and public keys are connected as $(SK_i)^2 = PK_i \mod N$, while public keys are generated via some PRNG from the fixed master key $K$ and number $i$. Only Alice can produce private keys since she knows prime factors $p$ and $q$, where $N=pq$. 

Another accepted solution was proposed by Ivan Ioganson, Zhan-Mishel Dakuo and Andrei Golovanov from Saint Petersburg ITMO University (Russia). Ideas of ID-based signature scheme are used in it. Public and private keys are generated from the corresponding master keys $PK_0$ and $SK_0$. The principles of Diffie-Hellman protocol on finite groups are applied. Namely, private keys are generated as $SK_i=SK_0 * H(i)$, whereas public keys used by the server are combinations of  $PK_0 = SK_0 * P$ and numbers $i$, where $P$ is a generator element of the group. It is hard to recover $SK_i$ by information on the server and from $SK_1,\ldots,SK_{i-1}, SK_{i+1},\ldots, SK_n$ if the hash function $h$ is of a good cryptographic quality.

\subsection{Problem ``CP Problem''}

\subsubsection{Formulation}
\hypertarget{pr-cpproblem}{}

	Let $\mathbb{G} = \langle g \rangle$ be a group of prime order $q$, $\kappa$ is the bit length of~$q$. Let us consider two known modifications of the discrete logarithm problem over $\mathbb{G}$, namely, $s$-DLOG problem and $\ell$-OMDL problem. Both of them are believed to be difficult.

\bigskip
	
\noindent \textbf{$s$-DLOG problem} (with parameter $s \in \mathbb{N}$)
	
\begin{tabular}{ll}	
		\underline{Unknown values}:    & $x$ is chosen uniformly at random from $\mathbb{Z}_q^*$.\\
		
		\underline{Known values}:      & $g^x, g^{x^2}, \ldots, g^{x^s}$.\\
		
		\underline{Access to oracles}: & no.\\
		
		\underline{The task}:          & to find $x$.\\
\end{tabular}
	
	~\\
	
\noindent 	\textbf{$\ell$-OMDL (One-More Discrete Log) problem} (with parameter $\ell \in \mathbb{N}$)

\begin{tabular}{lp{12cm}}		
		\underline{Unknown values}:  & $x_1, x_2, \ldots, x_{\ell +1}$ are chosen uniformly at random from $\mathbb{Z}_q^*$.\\
		
		\underline{Known values}:    & $g^{x_1}, g^{x_2}, \ldots, g^{x_{\ell +1}}$.\\
		
		\underline{Access to oracles}:& at most $\ell$ queries to $O_1$ that on input $y \in \mathbb{G}$ returns $x$ \\
		 & such that $g^x = y$. \\
		\underline{The task}:         & to find $x_1, x_2, \ldots, x_{\ell +1}$.\\
\end{tabular}

		~\\

	Consider another one problem that is close to the $s$-DLOG and $\ell$-OMDL problems:
		~\\

\noindent 	\textbf{$(k, t)$-CP (Chaum---Pedersen) problem} (with parameters $k, t \in \mathbb{N}$)

\begin{tabular}{lp{13.5cm}}		
	\underline{Unknown values}: & $x_1, x_2, \ldots, x_{t +1}$ are chosen uniformly at random from $\mathbb{Z}_q^*$.\\
	
	\underline{Known values}: & $g^{x_1}, g^{x_2}, \ldots, g^{x_{t +1}}$.\\
	
	\underline{Access to oracles}: & at most $k$ queries to $O_1$ that on input $(i, z) \in \{1, \ldots, t+1\}\times \mathbb{G}$ \\
	 &  returns $z^{x_i}$, and at most $t$ queries to $O_2$ that on input \\
    &  $(\alpha_1, \ldots, \alpha_{t+1}) \in  \mathbb{Z}_q^{t+1}$ returns ${\alpha_1 x_1 + \ldots + \alpha_{t+1} x_{t+1}}$. \\
	\underline{The task}: & to find $x_1, x_2, \ldots, x_{t+1}$.\\
\end{tabular}	

	~\\
	
	It is easy to see that if there exists a polynomial (by $\kappa$) algorithm that solves the $s$-DLOG problem, then there exists a polynomial algorithm that solves the $(s-1, t)$-CP problem for any $t \in \mathbb{N}$.

	\underline{\color{blue}{\bf Problem for a special prize!}} Prove or disprove the following conjecture: if there exists a polynomial algorithm that solves $(k, t)$-CP problem, then there exists a polynomial algorithm that solves at least one of the $s$-DLOG and $\ell$-OMDL problems, where $k, t, s, \ell$ are upper bounded by polynomial of $\kappa$.

\subsubsection{Solution}

Unfortunately, there were no any advances on solving this problem among participants, so, this conjecture is still open.

\subsection{Problem ``Interpolation with Errors''}

\subsubsection{Formulation}
\hypertarget{pr-interpolation}{}

\newcommand\F{\mathbb{F}}
\newcommand\Z{\mathbb{Z}}
\newcommand\Zn{\Z_n}
\renewcommand\d{16}
\newcommand\dm{15}

\newcommand\m{324}
\newcommand\mm{125}
\newcommand\mmm{109}
\newcommand\me{90}
\newcommand\err{35}

Let $n = 2022$ and let $\Zn$ be the ring of integers modulo $n$. Given $x_i, y_i \in \Zn$ for $i \in \{1,\ldots,\m\}$, find monic polynomials
\begin{align*}
f(x) &= x^{\d} + \alpha_{\dm} x^{\dm} + \ldots + \alpha_1 x + \alpha_0,\\
g(x) &= x^{\d} + \beta_{\dm} x^{\dm} + \ldots + \beta_1 x + \beta_0
\end{align*} of degree $d=\d$ and coefficients from $\Zn$ such that the relation \[
y_i
= \frac%
{f(x_i)}{g(x_i)}
= \frac%
{x_i^{\d} + \alpha_{\dm} x_i^{\dm} + \ldots + \alpha_1 x_i + \alpha_0}%
{x_i^{\d} + \beta_{\dm} x_i^{\dm} + \ldots + \beta_1 x_i + \beta_0}
\]
holds for at least $\me$ of the indices $i \in \{1,\ldots,\m\}$.

\bigskip

\noindent \textbf{Note.} The coefficients $\beta_0, \ldots, \beta_{\dm}$ are such that the denominator of the above fraction is invertible for all possible values of $x_i \in \Zn$. It can be assumed that they are sampled uniformly at random from all such sets of values. Furthermore, the positions and error values can be also assumed to be sampled uniformly at random.

\bigskip


 The attachment (see \cite{INTWITHERR}) contains a CSV file with $\m$ triplets $(i, x_i, y_i)$.

\subsubsection{Solution}

First, note that $n = 2022 = 2\cdot 3 \cdot 337$. Therefore, the problem can be solved for moduli $2,3,337$ independently, and then recovered using the Chinese Remainder Theorem (CRT). Furthermore, for moduli 2 and 3, there are only a few possible polynomials (modulo the relations $x^2=x$ modulo 2 and $x^3=x$ modulo 3). The best candidate polynomial modulo 6 (ignoring equivalent forms) satisfies $\mm$ / $\m$ values $x_i, y_i$, while the next best one does only $\mmm$ / $\m$. Note that the expected value is $\me + (\m-\me)/6 = 129$ ($\me$ correct ones and 1/6 wrong pairs satisfying the relation modulo 6 by chance), so that it is safe to assume that the best one is correct. We can now consider the problem modulo 337, where we know that the $\me$ correct pairs must be among those $\mm$ correct pairs observed modulo 6. Denote the set of those $\mm$ remaining indices by $I$.

Note that the relation can be rewritten as \[
y_i\cdot g(x_i) - f(x_i) = 0,
\]
or, more explicitly, 
\begin{align}
\label{eq:rel3}
\Big( y_i\cdot \sum_{j=0}^{\dm} \beta_i x_i^j \Big) - \Big(\sum_{j=0}^{\dm} \alpha_i x_i^j\Big) + \Big(y_i x_i^d - x_i^d\Big) = 0.
\end{align}
The target problem can now be formulated as the problem of decoding a linear code over the finite field $GF(337)$. Indeed, let the generator matrix $G$ be given by columns \[
	(
	-1, -x_i, -x_i^2, \ldots, -x_i^{\dm},~~
	y_i, y_i\cdot x_i, y_i\cdot x_i^2, \ldots, y_i\cdot x_i^{\dm})
\] 
for all chosen indexes $i \in I$,
let the target vector $v$ be given by \[
	v = (y_i x_i^d - x_i^d)_{i \in I},
\]
and consider the ``solution'' vector \[
	s = (\alpha_0, \ldots, \alpha_{\dm}, \beta_0, \ldots, \beta_{\dm}).
\]
It easy easy to verify that the codeword $s\times G$ differs from $-v$ in at most $\mm-\me$ places, i.e., has at most $\err$ errors. Indeed, the vector $s\times G$ compute the contribution of the first two clauses of Equation~\eqref{eq:rel3}, whereas $v$ defines the third clause, and the three clauses sum to zero on correct data pairs.
Note that $G$ defines a $[\mm, 32]$ code, i.e., a 32-dimensional code of length $\mm$. A random such code has expected minimum distance about 82 (given by the Gilbert-Varshamov bound), so that the solution (with the error $\err$ less than half of the distance) should likely be unique (modulo 337).

A very basic yet efficient method for linear code decoding is the so-called ``pooled Gauss'' method: choosing $k=32$ random coordinates of the code and assuming that they are error-free, allowing to recover full codeword by solving a linear system. Alternatively, SageMath includes an implementation of the Lee-Brickell method, which is slightly faster. The decoding should take less than 30 minutes using the basic method.

\textbf{Remark:} due to equivalent polynomial fractions modulo 2 and modulo 3, the overall solution is not unique (but there are only a few candidates).

\subsection{Problem ``HAS01''}

\subsubsection{Formulation}
\hypertarget{pr-HAS01}{}

Bob is a beginner cryptographer. He read an article about the new hash function HAS01 (see a description in \cite{HAS01}).
Bob decided to implement the HAS01 function in order to use it for checking the integrity of messages being forwarded. However, he was inattentive and made a mistake during the implementation. In the function $f_1$, he did not notice the sign <<'>> in the variable  $a$ and used the following set of formulas:

\smallskip
\begin{algorithmic}
  \FOR {$i = 0 \text{ to } 7$}
  \FOR {$j = 0 \text{ to } 6$}
  \STATE $a_{(i+1)\bmod 8,j} \leftarrow \text{SBox}(((a_{i,j}\oplus a_{(i+1)\bmod 8,j}) \ll 3) \oplus ((a_{i,j+1} \oplus a_{(i+1)\bmod 8,j+1}) \gg 5))$
  \ENDFOR
  \STATE $a_{(i+1)\bmod 8,7} \leftarrow \text{SBox}(((a_{i,7}\oplus a_{(i+1)\bmod 8,7}) \ll 3) \oplus ((a_{i,0} \oplus a_{(i+1)\bmod 8,0}) \gg 5) \oplus 7)$
    \ENDFOR
\end{algorithmic}

\begin{itemize}
\item[{\bf Q1}]
 Prove that Bob's version of the hash function is cryptographically weak.
\item[{\bf Q2}]
 Find a collision to the following message (given in hexadecimal format): \\
 {\tt 316520393820336220323620343720316320373820386520}.
\end{itemize}

\smallskip
The test set value for the original HAS01 hash function is given in \cite{HAS01_1}.

The test set value for Bob's implementation is given in \cite{HAS01_2}.

\subsubsection{Solution}

{\bf Q1} In the case where Bob makes a mistake and uses formulas with recursion, it turns out that for each first byte of the string (a00, a10, a20, a30, a40, a50, a60, a70), the most significant three bits do not affect the formation of the digest. Therefore, the function is not collision resistant, making it easy to pick up a number of different values that produce the same hash value.

{\bf Q2} According to the formulas, the most significant three bits for the first byte of each string do not affect the formation of the hash value. However, the original message fills only the first three rows of the original matrix. Therefore, changing the upper three bits in bytes a00, a10, a20 will allow you to get the same hash values. Hence, for a given value {\tt 316520393820336220323620343720316320373820386520}, you can get $2^9-1=511$ collisions.

For example:

{\tt 316520393820336220323620343720316320373820386520;}

{\tt F16520393820336220323620343720316320373820386520;}

{\tt F165203938203362E0323620343720316320373820386520;}

{\tt 31652039382033622032362034372031E320373820386520;}

and so on.

It should be noted, that most of those participants who tried to solve this problem were able to get the correct answer and determine the collision. Separately, it is worth noting that the team of Mikhail Borodin, Vitaly Kiryukhin and Andrey Rybkin (Russia) not only answered the questions of the task correctly, but also considered the issues of a possible vulnerability for the HAS01-512 algorithm.

\subsection{Problem ``Weaknesses of the PHIGFS''}

\subsubsection{Formulation}
\hypertarget{pr-PHIGFS}{}

A young cryptographer Philip designs a family of lightweight block ciphers based on a 4-line type-2 Generalized Feistel scheme (GFS) with better diffusion effect.

Its block is divided into four $m$-bit subblocks, $m \geqslant 1$. For better diffusion effect, Philip decides to use a $(4 \times 4)$-matrix $A$ over $\mathbb{F}_{2^m}$ instead of a standard subblocks shift register in each round. The family ${\rm{PHIGFS}}_{\ell}(A,b)$ is parameterized by a non-linear permutation $b\colon{\mathbb{F}_{2^m}} \to {\mathbb{F}_{2^m}}$, the matrix $A$ and the number of rounds $\ell \geqslant 1$. The one-round keyed transformation of ${\rm{PHIGFS}}_{\ell}(A,b)$ is a permutation ${g_k}$ on $\mathbb{F}_{2^m}^4$ defined as:
$${g_k}({x_3},{x_2},{x_1},{x_0}) = A\cdot{\left( {{x_3},{x_2} \oplus b({x_3} \oplus {k_1}),{x_1},{x_0} \oplus b({x_1} \oplus {k_0})} \right)^T},$$
where ${x_0},{x_1},{x_2},{x_3} \in {\mathbb{F}_{2^m}}$, $k = ({k_1},{k_0})$ is a $2m$-bit round key, ${k_0},{k_1} \in {\mathbb{F}_{2^m}}$.

The $\ell$-round encryption function $f_{{k^{(1)},\ldots,k^{(\ell)}} }\colon \mathbb{F}_{2^m}^4 \to \mathbb{F}_{2^m}^4$ under a key $({k^{(1)}},\ldots,{k^{(\ell)}}) \in \mathbb{F}_{2^m}^\ell$ is given by
$${f_{ {{k^{(1)}},\ldots,{k^{(\ell)}}} }}({\bf{x}}) = g_{k^{(\ell)}} \ldots g_{k^{(1)}}({\bf{x}}) \text{ for all } {\bf{x}} \in {\mathbb{F}}_{2^m}^4.$$

For effective implementation and security, Philip chooses two binary matrices ${A'}, {A''}$ with the maximum branch number among all binary matrices of size 4, where
$${A'} = \left( {\begin{array}{*{20}{c}}
1&1&0&1\\
1&0&1&1\\
0&1&1&1\\
1&1&1&0
\end{array}} \right),
\;
{A''} = \left( {\begin{array}{*{20}{c}}
0&1&1&1\\
1&1&1&0\\
1&1&0&1\\
1&0&1&1
\end{array}} \right).$$

For approval, he shows the cipher to his friend Antony who claims that ${A'},{A''}$ are bad choices because ciphers ${{\rm{PHIGFS}}_\ell}(A',b)$, ${{\rm{PHIGFS}}_\ell}(A'',b)$ are insecure against distinguisher attacks for all $b\colon {\mathbb{F}_{2^m}} \to {\mathbb{F}_{2^m}}$, $\ell \geqslant 1$.

\bigskip

Help Philip to analyze the cipher ${{\rm{PHIGFS}}_\ell}(A,b)$.
Namely, for any $b\colon{\mathbb{F}_{2^m}} \to {\mathbb{F}_{2^m}}$ and any $\ell \geqslant 1$, show that ${\rm{PHIGFS}}_{\ell}(A,b)$ has
  \begin{enumerate}[noitemsep]
    \item[{\bf (a)}]{$\ell$-round differential sets with probability 1; }
    \item[{\bf (b)}]{$\ell$-round impossible differential sets;}
  \end{enumerate}
for the following cases:
    {\bf Q1} $A = A'$;
and
  {\bf Q2} $A = A''$.
In each case, construct these nontrivial differential sets and prove the corresponding~property.

\noindent{\bf Remark.} Let us recall the following definitions.

\begin{itemize}
\item
Let $\delta ,\varepsilon  \in {\mathbb{F}_{{2^n}}}$ be fixed nonzero input and output differences. The \textit{differential probability} of $s\colon{\mathbb{F}_{2^n}} \to {\mathbb{F}_{2^n}}$ is defined as
$${p_{\delta ,\varepsilon }}(s) = {2^{ - n}} \cdot \left|{\left\{ {\alpha  \in {\mathbb{F}_{{2^n}}}|s(\alpha  \oplus \delta ) \oplus s(\alpha ) = \varepsilon } \right\}}\right|.$$

\item If $s\colon{\mathbb{F}_{{2^n}}} \times K \to {\mathbb{F}_{{2^n}}}$ depends on a key space $K$, then the \textit{differential probability} of $s$ is defined as
$${p_{\delta ,\varepsilon }}(s) = {\left| K \right|^{ - 1}}\sum\limits_{k \in K} {{p_{\delta ,\varepsilon }}({s_k})} ,$$
where $s({x},k) = {s_k}({x})$, ${x} \in {\mathbb{F}_{{2^n}}}$, $k \in K$.

\item Let $\Omega,\Delta \subseteq  {{\mathbb{F}}_{2^n}} \backslash \{0\} $ and $\Omega,\Delta$ are nonempty. If ${p_{\delta ,\varepsilon }}(s) = 0$ for any $\delta\in \Omega, \; \varepsilon \in \Delta$, then $(\Omega, \Delta)$ are \textit{impossible differential sets}. But if
$$\sum_{\delta\in \Omega, \varepsilon \in \Delta} {p_{\delta ,\varepsilon }}(s) =1, $$
then $(\Omega, \Delta)$ are \textit{differential sets with probability} 1.
We call $(\Omega, \Delta)$ trivial (impossible) differential sets if $\Omega \in \{\emptyset, {{\mathbb{F}}_{2^n}} \backslash \{0\} \} $ or $\Delta \in \{\emptyset, {{\mathbb{F}}_{2^n}} \backslash \{0\} \}$.
\end{itemize}

\subsubsection{Solution}
Let $\delta ,\varepsilon  \in {\mathbb{F}_{{2^n}}}$ be fixed nonzero input and output differences. The differential probability of $s\colon{\mathbb{F}_{2^n}} \to {\mathbb{F}_{2^n}}$ is defined as
$${p_{\delta ,\varepsilon }}(s) = {2^{ - n}} \cdot \left|{\left\{ {\alpha  \in {\mathbb{F}_{{2^n}}}|s(\alpha  \oplus \delta ) \oplus s(\alpha ) = \varepsilon } \right\}}\right|.$$
If $s\colon{\mathbb{F}_{{2^n}}} \times K \to {\mathbb{F}_{{2^n}}}$ depends on a key space $K$ then the differential probability of $s$ is defined as
$${p_{\delta ,\varepsilon }}(s) = {\left| K \right|^{ - 1}}\sum\limits_{k \in K} {{p_{\delta ,\varepsilon }}({s_k})} ,$$
where $s({x},k) = {s_k}({x})$, ${x} \in {\mathbb{F}_{{2^n}}}$, $k \in K$. In that case the pair $(\delta ,\varepsilon )$ represents a differential denoted by $\delta { \longrightarrow ^s}\varepsilon $.

For the $l$-round encryption function $f$, we will sometimes write $\delta { \longrightarrow _l}\varepsilon $ to emphasize the number of rounds $l$ instead of $\delta { \longrightarrow ^f}\varepsilon $.

For $\delta  \in {\mathbb{F}_{{2^m}}},b:{\mathbb{F}_{{2^m}}} \to {\mathbb{F}_{{2^m}}}$, we denote
$${\Delta _\delta }(b) = \left\{ {b(\alpha  \oplus \delta ) \oplus b(\alpha )|\alpha  \in {\mathbb{F}_{{2^m}}}} \right\}.$$

Note that ${g_k}$ consists of a transformation ${v_k}:\mathbb{F}_{{2^m}}^4 \to \mathbb{F}_{{2^m}}^4$ and the matrix $a$ over ${\mathbb{F}_{{2^m}}}$ , where
$${v_k}({x_3},{x_2},{x_1},{x_0}) = \left( {{x_3},{x_2} \oplus b({x_3} \oplus {k_1}),{x_1},{x_0} \oplus b({x_1} \oplus {k_0})} \right),$$
$${g_k}({\bf{x}}) = a{\left( {{v_k}({\bf{x}})} \right)^T},{\bf{x}} \in \mathbb{F}_{{2^m}}^4.$$

\textbf{Case I. $a = {a_1}$}.

Let $\varepsilon  \in {\mathbb{F}_{2^m}}$, 
$$W(\varepsilon ) = \left\{ {({\alpha _3},{\alpha _2},{\alpha _1},{\alpha _0}) \in \mathbb{F}_{{2^m}}^4|{\alpha _3} \oplus {\alpha _1} = \varepsilon } \right\}\backslash \left\{ {\left( {0,0,0,0} \right)} \right\}.$$

\textbf{Theorem} 1. Let $l$ be any positive integer, $\varepsilon  \in {\mathbb{F}_{2^m}}$. Then
$l$-round differential sets $W(\varepsilon ){\longrightarrow_{l}} W(\varepsilon )$ of the ${\rm{PHIGFS}}_{l}({a_1},b)$ hold with probability 1.

\textbf{Proof.} 

 Note that for any $({x_3},{x_2},{x_1},{x_0}) \in \mathbb{F}_{{2^m}}^4$ we have the following equality
$$a_{1}{({x_3},{x_2},{x_1},{x_0})^T} = {\left( {{x_3} \oplus {x_2} \oplus {x_0},{x_3} \oplus {x_1} \oplus {x_0},{x_2} \oplus {x_1} \oplus {x_0},{x_3} \oplus {x_2} \oplus {x_1}} \right)^T}.$$

Let us consider any nonzero $(\delta ,\lambda ,\omega ) \in \mathbb{F}_{2^m}^3$ and any round key $k \in \mathbb{F}_{{2^m}}^2$. 

Note that $v_k$ maps a difference
 $$\left( {\delta ,\lambda ,\delta  \oplus \varepsilon ,\omega } \right) \in W(\varepsilon) \text{ to a difference } \left( {\delta ,{\lambda ^{(1)}},\delta  \oplus \varepsilon ,{\omega ^{(1)}}} \right) \in W(\varepsilon) $$
 for any
 $${\lambda ^{(1)}} \in {\Delta _\delta }(b) \oplus \lambda, \; {\omega ^{(1)}} \in {\Delta _{\delta  \oplus \varepsilon }}(b) \oplus \omega.$$
Then 
$$a_{1} \left( {\delta ,{\lambda ^{(1)}},\delta  \oplus \varepsilon ,{\omega ^{(1)}}} \right) = \left( {{\omega ^{(1)}} \oplus \delta  \oplus {\lambda ^{(1)}},{\omega ^{(1)}} \oplus \varepsilon ,{\omega ^{(1)}} \oplus \delta  \oplus {\lambda ^{(1)}} \oplus \varepsilon ,{\lambda ^{(1)}} \oplus \varepsilon } \right).$$
Thus, $g_k$ encrypts the difference 
$$\left( {\delta ,\lambda ,\delta  \oplus \varepsilon ,\omega } \right)\in W(\varepsilon) \text{ to the difference }
\left( {{\delta ^{(1)}},{\lambda ^{(2)}},{\delta ^{(1)}} \oplus \varepsilon ,{\omega ^{(2)}}} \right)\in W(\varepsilon) ,$$ 
where 
 $${\delta ^{(1)}} = {\lambda ^{(1)}} \oplus \delta  \oplus {\omega ^{(1)}}, \; {\lambda ^{(2)}} = {\omega ^{(1)}} \oplus \varepsilon, \; {\omega ^{(2)}} = {\lambda ^{(1)}} \oplus \varepsilon.$$
Therefore, 
$$P\left\{ W(\varepsilon) { \longrightarrow ^g} W(\varepsilon) \right\} = 1.$$

By induction on the number of rounds $l$, we can straightforwardly get
 $$P\left\{ W(\varepsilon) { \longrightarrow_{l}} W(\varepsilon) \right\} = 1.$$
 $\Box$

\textbf{Corollary 1.} For any number of rounds $l \ge 1$, {($W(\varepsilon ), W(\delta ))$ are a pair of impossible $l$-round differential sets for any different $\varepsilon ,\delta  \in {\mathbb{F}_{2^m}}$.}

	The proof follows from Theorem 1.  $\Box$

\textbf{Case II. $a = {a_2}$}.

Let
$$W = \left\{ {(0,\delta ,\delta ,\theta )|(\delta ,\theta ) \in \mathbb{F}_{{2^m}}^2\backslash \left\{ {(0,0)} \right\}} \right\}.$$

\textbf{Theorem} 2.
Let $l$ be any positive integer, $\varepsilon  \in {\mathbb{F}_{2^m}}$. Then
$l$-round differential sets $W{\longrightarrow_{l}} W$ of the ${\rm{PHIGFS}}_{l}({a_2},b)$ holds with probability 1.

\textbf{Proof.}
Note that for any $({x_3},{x_2},{x_1},{x_0}) \in \mathbb{F}_{{2^m}}^4$ we have
$$a_2{({x_3},{x_2},{x_1},{x_0})^T} = {\left( {{x_3} \oplus {x_2} \oplus {x_1},{x_3} \oplus {x_2} \oplus {x_0},{x_3} \oplus {x_1} \oplus {x_0},{x_2} \oplus {x_1} \oplus {x_0}} \right)^T}.$$

Let us consider any nonzero $(\delta ,\theta ) \in \mathbb{F}_{2^m}^2$ and any round key $k \in \mathbb{F}_{{2^m}}^2$.

Note that $v_k$ maps a difference
 $$\left( {0,\delta ,\delta ,\theta } \right) \in W \text{ to a difference } \left( {0,\delta ,\delta ,{\theta ^{(1)}}} \right) \in W $$
 for any ${\theta ^{(1)}} \in {\Delta _\delta }(b) \oplus \gamma.$
Then
$$a_{2} \left( {0,\delta ,\delta ,{\theta ^{(1)}}} \right) =\left( {0,{\theta ^{(1)}} \oplus \delta ,{\theta ^{(1)}} \oplus \delta ,{\theta ^{(1)}}} \right).$$
Thus, $g_k$ encrypts the difference
$$ \left( {0,\delta ,\delta ,\theta } \right) \in W \text{ to the difference }
\left( {0,{\delta ^{(1)}},{\delta ^{(1)}},{\theta ^{(1)}}} \right) \in W,$$
where ${\delta ^{(1)}} = {\theta ^{(1)}} \oplus \delta $.
Therefore,
$$P\left\{ W {\longrightarrow ^g} W \right\} = 1.$$
By induction on the number of rounds $l$, we can straightforwardly get
 $$P\left\{ W { \longrightarrow_{l}} W \right\} = 1.$$
 $\Box$

\textbf{Corollary.} For any the number of rounds $l \ge 1$,  $(W, W')$ are a pair of impossible $l$-round differential sets for any  $W' \subseteq \mathbb{F}_{2^m}^4\backslash( W \cup \{ 0 \} )$.

	The proof follows from Theorem 2.  $\Box$

\bigskip
We would like to mention the solution of Gabriel Tulba-Lecu, Ioan Dragomir and Mircea-Costin Preoteasa (Romania).

\subsection{Problem ``Super dependent S-box''}

\subsubsection{Formulation}
\hypertarget{pr-SBOX}{}

Harry wants to find a super dependent S-box for his new cipher. He decided to use a permutation that is strictly connected with every of its variables. He tries to estimate the number of such permutations.

A vectorial Boolean function  $F(x)=\left(f_1(x),f_2(x),\ldots,f_n(x)\right)$, where $x\in\mathbb{F}_2^n$, is a~\textit{permutation} on~$\mathbb{F}_2^n$ if it is a one-to-one mapping on the set~$\mathbb{F}_2^n$. Its coordinate
function~$f_k(x)$ (that is a Boolean function from $\mathbb{F}_2^n$ to $\mathbb{F}_2$),~\textit{essentially depends} on the variable~$x_j$ if there exist  values $b_1,b_2,\ldots,b_{j-1},b_{j+1},\ldots,b_n \in\mathbb{F}_2$ such that
\begin{equation*}
f_k\left(b_1,b_2,\ldots,b_{j-1},0,b_{j+1},\ldots,b_n\right)\ne f_k\left(b_1,b_2,\ldots,b_{j-1},1,b_{j+1},\ldots,b_n\right).
\end{equation*}
In other words, the essential dependence on the variable $x_j$ of a function $f$ means the presence of $x_j$ in the algebraic normal form of $f$ (the unique representation of a function in the basis of binary operations AND, XOR, and constants 0 and 1).

\medskip
{\bf An example.} Let $n=3$. Then the Boolean function $f(x_1, x_2, x_3) = x_1x_2\oplus x_3$ essentially depends on all its variables; but $g(x_1, x_2, x_3) = x_1x_2\oplus x_2\oplus 1$ essentially depends only on $x_1$ and $x_2$. 

\bigskip

{\bf The problem.} Find the number of permutations on~$\mathbb{F}_2^n$ such that all their coordinate functions essentially depend on all~$n$ variables, namely
\begin{itemize}
\item[{\bf Q1}] Solve the problem for $n = 2,3$.
\item[{\bf Q2}] \underline{\color{blue}{\bf Problem for a special prize!}} Solve the problem for arbitrary $n$.
\end{itemize}

\bigskip

\subsubsection{Solution}

Let us denote the number of super-dependent S-boxes in $n$ variables by $\sdsbox{n}$. 
We can represent $F$ as $F(x) = (f_1(x), \ldots, f_n(x))$, where $x \in \mathbb{F}_2^n$ and $f_1, \ldots, f_n$ are Boolean functions in $n$ variables (i.e. functions of the form $\mathbb{F}_2^n \to \mathbb{F}_2$). Recall that $F$ is a permutation if and only if any its component function $b_1 f_1(x) \oplus\ldots\oplus b_n f_n(x)$, $b \in \mathbb{F}_2^n \setminus \{0\}$, is balanced (i.e. it takes zero and one in the same number of arguments).

The most of solutions provided by the participants contain an answer for Q1. As a rule, an exhaustive search was used. The correct answer for Q1 is the following: $\sdsbox{2} = 0$ and $\sdsbox{3} = 24576$. At the same time, some progress has been made on Q2. A short description of these results is bellow.

The team of Mikhail Kudinov, Denis Nabokov and Alexey Zelenetskiy (Russia) used the inclusion-exclusion principle and provided lower and upper bounds for $\sdsbox{n}$. Their ideas were the following.
Let $\sdfunc{k}$ be the set of functions $f : \mathbb{F}_2^k \to \mathbb{F}_2$ that essentially depend on all its variables $x_1, \ldots, x_k$. Then,
$$
	|\sdfunc{n}| = C_{2^n}^{2^{n - 1}} - \sum_{k = 0}^{n - 1}C_n^k|\sdfunc{k}|,
$$
where $C_n^k$ is a binomial coefficient. Next, let us define for any $i \in \{1, \ldots, n\}$ the sets
$$
	A_i = \{ \text{a permutation } F(x) = (f_1(x), \ldots, f_n(x)) \text{ on } \mathbb{F}_2^n : f_i \notin \sdfunc{n} \}.
$$
It means that the number of super-dependent S-boxes is the following:
$$
	\sdsbox{n} = 2^n! - |A_1 \cup \ldots \cup A_n|.
$$
It is not difficult to see that $|A_{i_1} \cap \ldots \cap A_{i_k}| = |A_1 \cap \ldots \cap A_k|$ for any $1 \leq k \leq n$ and any $k$-element set $\{i_1, \ldots, i_k\} \subseteq \{1, \ldots, n\}$. 
The inclusion-exclusion principle gives us that
$$
	\sdsbox{n} = 2^n! + \sum_{k = 1}^{n}{(-1)^k C_n^k |A_1 \cap \ldots \cap A_k|}.
$$
The cardinalities of intersections can be calculated in the following way:
$$
	|A_1 \cap \ldots \cap A_k| = 2^n! \frac{d(n,k)}{\prod_{i = 0}^{k - 1}(C_{2^{n - i}}^{2^{n - i - 1}})^{2^i}}, 
$$
where $d(n, k)$ is the number of tuples $(f_1, \ldots, f_k)$ consists of Boolean functions in $n$ variables such that $f_1, \ldots, f_k \notin \sdfunc{n}$ and $b_1 f_1 \oplus \ldots \oplus b_k f_k$ is balanced for any $b \in \mathbb{F}_2^k \setminus \{0\}$. It is not easy to calculate $d(n, k)$. However, there is a trivial estimation $d(n, k) \geq C_{2^{n-1}}^{2^{n - 2}}$. Also,
$$
	|A_1| = 2^n! \frac{C_{2^{n}}^{2^{n - 1}} - |\sdfunc{n}|}{C_{2^{n}}^{2^{n - 1}}}.
$$
This can be used to estimate $\sdsbox{n}$:
$$
	2^n! - n|A_1| \leq \sdsbox{n} \leq 2^n! - |A_1|.
$$

The team of Stepan Davydov, Anastasiia Chichaeva and Kirill Tsaregorodtsev (Russia) proposed interesting ideas as well. They noticed that $2^n \mid \sdsbox{n}$, implemented Monte-Carlo simulations for $n = 4$ and $n = 5$ and showed that $\lim_{n \to \infty} \frac{\sdsbox{n}}{2^n!} = 1$. Also, the team pointed out a subclass of super-dependent S-boxes such that even component functions of its representatives essentially depend on all its variables.

The team of Mikhail Borodin, Vitaly Kiryukhin and Andrey Rybkin (Russia) calculated that $S(4) = 19344102217728 = 24 \cdot 16 \cdot 50375266192$. They used that the addition to a super-dependent S-box in $n$ variables of any binary vector from $\mathbb{F}_2^n$ and rearranging its output bits provided a super-dependent S-box as well. In other words, $n! \cdot 2^n \mid S(n)$ holds. Note that some other participants mentioned such kind of classifications (for instance, in the solution above). However, the team most successfully exploited this fact. 

\bigskip

\bigskip

\subsection{Problem ``Quantum entanglement ''}

\subsubsection{Formulation}
\hypertarget{pr-entanglement}{}

	The Nobel Prize in Physics in 2022 was awarded to researchers who experimentally investigated quantum~\textit{entanglement}. One of their studies was devoted to a Greenberger--Horne--Zeilinger state $\ket{GHZ}=\frac{1}{\sqrt{2}}(\ket{000}+\ket{111})$,
	which is an entangled state of three qubits. This state can be created using the following quantum circuit:
	\begin{center}
		\begin{quantikz}[ampersand replacement=\&]
			\lstick{$\ket{0}$} \& \gate{H} \& \ctrl{1} \& \ctrl{2} \& \qw \\
			\lstick{$\ket{0}$} \& \qw \& \targ{} \& \qw \& \qw  \\
			\lstick{$\ket{0}$} \& \qw \& \qw \& \targ{} \& \qw
		\end{quantikz}
	\end{center}

	After the measurement, the probability to find the system described by~$\ket{GHZ}$ in the state~$\ket{000}$ or in the state~$\ket{111}$ is equal to~$1/2$.
	
	When we make measurements in quantum physics, we are able to make~{\it post-selection}. For example, if we post-select the events when the first qubit was in state~$\ket{0}$, the second and the third qubits will also be found in the state~$\ket{0}$ for sure, this is actually what entanglement means. We also see that the post-selection destroys entanglement of two remaining qubits.
	
\begin{itemize}
\item[{\bf Q1}] But what will happen, if we post-select the events when the 1st qubit is in the Hadamard state~$\ket{+}=\frac{1}{\sqrt{2}}(\ket{0}+\ket{1})$? How can we perform this kind of post-selection if the result of each measurement of a qubit state can be only~$0$ or~$1$ and we can only post-select these events? Will the two remaining qubits be entangled after post-selection? Design the circuit which will provide an answer.

\item[{\bf Q2}]  \underline{\color{blue}{\bf Problem for a special prize!}} There are two different classes of three-qubit entanglement. One of them is
\begin{equation*}
	\ket{GHZ}=\frac{1}{\sqrt{2}}(\ket{000}+\ket{111}),
\end{equation*}
  and the other is
\begin{equation*}
	\ket{W}=\frac{1}{\sqrt{3}}(\ket{001}+\ket{010}+\ket{100}).
\end{equation*}
Discuss the possible ideas how the difference between these states can be found with the usage of post-selection and measurement. Don't forget that you need to verify entanglement for both types of states!
\end{itemize}

{\small
\noindent{\bf Remark.} Let us briefly formulate the key points of quantum circuits.
A qubit is a two-level quantum mechanical system whose state $\ket{\psi}$ is the superposition of basis quantum states $\ket{0}$ and $\ket{1}$. The superposition is written as $\ket{\psi}=\alpha_0\ket{0}+\alpha_1\ket{1}$, where $\alpha_0$ and $\alpha_1$ are complex numbers, called amplitudes, that possess $\left|\alpha_0\right|^2+\left|\alpha_1\right|^2=1$. The amplitudes $\alpha_0$ and $\alpha_1$ have the following physical meaning: after the measurement of a qubit which has the state $\ket{\psi}$, it will be observed in the state $\ket{0}$ with probability $\left|\alpha_0\right|^2$ and in the state $\ket{1}$ with probability $\left|\alpha_1\right|^2$.
Note that we can measure qubit, initially given in the state~$\ket{\psi}=\alpha_0\ket{0}+\alpha_1\ket{1}$, in other basis, for example Hadamard basis~$\ket{+}=\frac{1}{\sqrt{2}}(\ket{0}+\ket{1})$ and~$\ket{-}=\frac{1}{\sqrt{2}}(\ket{0}+\ket{1})$. In order to do this, we consider the state in the form~$\ket{\psi}=\alpha_0'\ket{+}+\alpha_1'\ket{-}$, where complex amplitudes~$\alpha_0',\alpha_1'$ have the same physical meaning as $\alpha_0$ and $\alpha_1$. Then we can calculate the probability that the qubit will be in the state~$\ket{+}$ or~$\ket{-}$ after the measurement and consider the process of post-selection in this case.
In order to operate with multi-qubit systems, we consider the bilinear operation $\otimes:\ket{x},\ket{y}\rightarrow\ket{x}\otimes\ket{y}$ on $x,y\in\{0,1\}$ which is defined on pairs $\ket{x},\ket{y}$, and by bilinearity is expanded on the space of all linear combinations of $\ket{0}$ and $\ket{1}$. When we have two qubits in states $\ket{\psi}$ and $\ket{\varphi}$ correspondingly, the state of the whole system of these two qubits is
$
\ket{\psi}\otimes\ket{\varphi}.
$
In general, for two qubits we have
$
\ket{\psi}=\alpha_{00}{\ket{0}\otimes\ket{0}}+\alpha_{01}\ket{0}\otimes\ket{1}+\alpha_{10}\ket{1}\otimes\ket{0}+\alpha_{11}\ket{1}\otimes\ket{1}.
$
The physical meaning of complex numbers $\alpha_{ij}$ is the same as for one qubit, so we have the essential restriction $|\alpha_{00}|^2+|\alpha_{01}|^2+|\alpha_{10}|^2+|\alpha_{11}|^2=1$. We use more brief notation $\ket{a}\otimes\ket{b}\equiv\ket{ab}$.
By induction, this process is expanded on the case of three qubits and more.
Mathematically, the entanglement of $n$-qubits state means that we can not consider this state in the form $\ket{\psi}=\ket{\varphi_1}\otimes\ket{\varphi_2}$, where $\ket{\varphi_1}$ and $\ket{\varphi_2}$ are some states of~$m$ and~$n-m$ qubits, correspondingly.
In order to verify your circuits, you can use different quantum circuit simulators, for example, see \cite{quirk}.
}

\subsubsection{Solution}

The first question. 

The circuit for creation of the Greenberger--Horne--Zeilinger state $\ket{GHZ}$ is the following:
\begin{center}
	\begin{quantikz}[ampersand replacement=\&]
		\lstick{\ket{0}} \& \gate{H} \& \ctrl{1} \& \ctrl{2} \& \qw\rstick[wires=3]{$\frac{\ket{000}+\ket{111}}{\sqrt{2}}$} \\
		\lstick{\ket{0}} \& \qw \& \targ{} \& \qw \& \qw \\
		\lstick{\ket{0}} \& \qw \& \qw \& \targ{2} \& \qw
	\end{quantikz}
\end{center}

First, we need to post-select events when the first qubit is in the Hadamard state ${\ket{+}=\frac{1}{\sqrt{2}}(\ket{0}+\ket{1})}$. For this purpose, we make an Hadamard gate prior to the measurement of the first qubit. After this we perform a post-selection.

The state $\ket{GHZ}$ can be written as
\begin{equation*}
	\ket{GHZ}=\frac{\ket{000}+\ket{111}}{\sqrt{2}}=\ket{+}\frac{(\ket{00}+\ket{11})}{2\sqrt{2}}+\ket{-}\frac{(\ket{00}-\ket{11})}{2\sqrt{2}},
\end{equation*}
where~$\ket{\pm}=(\ket{0}\pm\ket{1})/\sqrt{2}$. It means that if we select the first qubit in the state $\ket{+}$, the other qubits will be in the entangled Bell state $\ket{\Phi^+}=\frac{1}{\sqrt{2}}(\ket{00}+\ket{11})$. This state can be detected using a CNOT gate followed by the Hadamard gate. The whole circuit is
\begin{center}
	\begin{quantikz}[ampersand replacement=\&]
		\lstick{\ket{0}} \& \gate{H} \& \ctrl{1} \& \ctrl{2} \& \gate{H} \& \push{\ket{0}\bra{0}} \& \qw \& \qw \& \qw \\
		\lstick{\ket{0}} \& \qw \& \targ{} \& \qw \& \qw \& \qw \& \ctrl{1} \& \gate{H} \& \qw \\
		\lstick{\ket{0}} \& \qw \& \qw \& \targ{2} \& \qw \& \qw \& \targ{} \& \qw \& \qw
	\end{quantikz}
\end{center}

The second question that supposed to be the open problem was solved during the Olympiad by the team of Viet-Sang Nguyen, Nhat Linh LE Tan and Phuong Hoa Nguyen (France). Here we provide the solution.

If we measure any qubit of the state~$\ket{GHZ}$ and known the result of the measurement, the state of two rest qubits immediately become known to us. Thus, the state of the whole system of 3 qubits is an entangled one. But the state of two rest qubits after the measurement of any qubit is separable.

When we measure the first qubit of the state~$\ket{W}$, the result is 0 with probability~$2/3$, and~$1/3$ for the result 1. When the state of the first qubit is measured 1, the system collapses to a separable state~$\ket{00}$ hence it is not entangled anymore. However, when the state of the first qubit is measured 0, the remaining two qubits become the maximally entangled state of two qubits. Given the measurement of one qubit as~$\ket{1}$, we can deduce the information of the other two because there is correlation in the information between qubits. Thus,~$\ket{W}$ is an entangled quantum state of three qubits.

Different from~$\ket{GHZ}$, measuring one qubit in~$\ket{W}$ creates an entangle state of two remaining qubit with probability~$2/3$. While being in~$\ket{GHZ}$, the system collapses to a separable state after measurement of any qubit.

The post-selection procedure for the state~$\ket{GHZ}$ was discussed in the first question, so the same technique can be applied for the state~$\ket{W}$. This state contains residual entanglement after measurement of a qubit, we can post-selection the third qubit in the state~$\ket{0}$ to attain the Bell state of the remaining qubits.
\begin{center}
	\begin{quantikz}[ampersand replacement=\&]
		\lstick{\ket{0}} \& \gate{R_y(\theta)} \& \ctrl{1} \& \qw \& \ctrl{1} \& \targ{} \& \push{\ket{0}\bra{0}} \& \qw \\
		\lstick{\ket{0}} \& \qw \& \gate{H} \& \ctrl{1} \& \targ{} \& \qw \& \qw \& \qw \\
		\lstick{\ket{0}} \& \qw \& \qw \& \targ{} \& \qw \& \qw \& \qw \& \qw
	\end{quantikz}
\end{center}
Here $R_y(\theta)$ gate is a single-qubit rotation through angle $\theta=2\arccos(1/\sqrt{3})$ (radians) around the $y$-axis.

The state~$\ket{W}$ has the following representation
\begin{align*}
	\ket{W}&=\frac{1}{\sqrt{3}}\big(\ket{001}+\ket{010}+\ket{100}\big)\\
	&=\frac{1}{\sqrt{6}}\big(\ket{00+}+\ket{01+}+\ket{10+}-\ket{00-}+\ket{01-}+\ket{10-}\big)
\end{align*}
If we can post-select the state $\ket{+}$ for the third qubit, we have:
\begin{equation*}
	\frac{1}{\sqrt{3}}\big(\ket{00+}+\ket{01+}+\ket{10+}\big)=\frac{1}{\sqrt{3}}\big(\ket{00}+\ket{01}+\ket{10}\big)\otimes\ket{+},
\end{equation*}
which is equivalent to a circuit with two entangled qubits similar to~$\ket{W}$ and a independent qubit in the state~$\ket{+}$. There is a correlation between~2 rest qubits in this system: if we measure~1 in one qubit, the other must be~0. Hence, we have an entanglement between~2 qubits.

\begin{center}
	\begin{quantikz}[ampersand replacement=\&]
		\lstick{\ket{0}} \& \gate{R_y(\theta)} \& \ctrl{1} \& \qw \& \ctrl{1} \& \targ{} \& \push{\ket{0}\bra{0}} \& \qw \\
		\lstick{\ket{0}} \& \qw \& \gate{H} \& \ctrl{1} \& \targ{} \& \qw \& \qw \& \qw \\
		\lstick{\ket{0}} \& \qw \& \qw \& \targ{} \& \qw \& \qw \& \qw \& \qw
	\end{quantikz}
\end{center}

The circuit for the system with third qubit in the state~$\ket{+}$ and 2 entangled qubits
\begin{center}
	\begin{quantikz}[ampersand replacement=\&]
		\lstick{\ket{0}} \& \gate{H} \& \qw \& \qw \& \qw \\
		\lstick{\ket{0}} \& \gate{R_y(\theta)} \& \ctrl{1} \& \targ{} \& \qw \\
		\lstick{\ket{0}} \& \qw \& \gate{H} \& \qw \& \qw
	\end{quantikz}
\end{center}

In conclusion, when measuring one qubit of the state $\ket{W}$, the state of the other two qubits are still entangled. But after the measurement of any qubit of the state~$\ket{GHZ}$, the states of the rest qubits become known. When post measuring Hadamard~$\ket{+}$ state, both~$\ket{W}$ and~$\ket{GHZ}$ states return outcome equivalent to a separate qubit in the state~$\ket{+}$ and a entangled state of two qubits.

We also would like to mention participants who made a progress in solution, that is the team of Gabriel Tulba-Lecu, Mircea-Costin Preoteasa and Ioan Dragomir (Romania), the team of Mikhail Kudinov, Denis Nabokov and Alexey Zelenetskiy (Russia), the team of Himanshu Sheoran, Gyumin Roh and Yo Iida (India, South Korea, Japan) and the team of
Donat Akos Koller, Csaba Kiss and Marton Marits (Hungary).

\section{Acknowledgement}

The authors are grateful to Andrey Nelyubin, Yuliya Maksimlyuk, Irina Khilchuk, Darya Zyubina and Valeria Kochetkova for useful discussions and various help. 
{The work of the first, second, third, fifth, seventh and eighth authors was supported by the Mathematical Center in Akademgorodok under the agreement No. 075-15-2022-282 with the Ministry of Science and Higher Education of the Russian Federation. The work of the ninth author was supported by the Kovalevskaya North-West Centre of Mathematical Research under the agreement No. 075-02-2023-934 with the Ministry of Science and Higher Education of the Russian Federation. The work is also supported by Novosibirsk State University and Kryptonite.}

\end{document}